\documentclass[11pt]{article}
\usepackage{amsmath, amssymb, amsfonts, amsthm}

\newtheorem{theorem}{Theorem}[section]
\newtheorem{proposition}[theorem]{Proposition}

\newtheorem{lemma}[theorem]{Lemma}

\newcommand{\rd}{{\rm d}}
\newcommand{\be}{\begin{equation}}
\newcommand{\ee}{\end{equation}}
\newcommand{\bey}{\begin{eqnarray}}
\newcommand{\eey}{\end{eqnarray}}

\newcommand{\PP}{{\mathbb P }}
\newcommand{\E}{{\mathbb E }}

\newcommand{\eps}{\varepsilon}

\newcommand{\ph}{\varphi}

\newcommand{\cU}{{\cal U}}

\newcommand{\bR}{{\mathbb R}}
\newcommand{\bC}{{\mathbb C}}
\newcommand{\bN}{{\mathbb N}}

\newcommand{\tr}{\mbox{Tr}}

\newcommand{\wt}{\widetilde}

\newcommand{\cF}{{\cal F}}
\newcommand{\cA}{{\cal A}}

\newcommand{\cD}{{\cal D}}

\newcommand{\cW}{{\cal W}}
\newcommand{\cK}{{\cal K}}
\newcommand{\cH}{{\cal H}}
\newcommand{\cL}{{\cal L}}
\newcommand{\cN}{{\cal N}}
\newcommand{\cO}{{\cal O}}

\input epsf

\newcommand{\donothing}[1]{}

\begin{document}

\title{A Central Limit Theorem in Many-Body Quantum Dynamics}

\author{G{\'e}rard Ben Arous, Kay Kirkpatrick\thanks{Partially supported by NSF grants OISE-0730136 and DMS-1106770 } and Benjamin Schlein\thanks{Partially supported by an ERC Starting Grant} \\ \\
Courant Institute of Mathematics, New York University \\ 251 Mercer Street, New York, NY 10012, USA \\ \\
Department of Mathematics, University of Illinois at Urbana-Champaign\\
1409 W. Green Street, Urbana, IL 61801, USA \\ \\
Institute of Applied Mathematics,\\ Endenicher Allee 60, 53115 Bonn, Germany}

\maketitle

\begin{abstract}
We study the many body quantum evolution of bosonic systems in the mean field limit. The dynamics is known to be well approximated by the Hartree equation. So far, the available results have the form of a law of large numbers. In this paper we go one step further and we show that the fluctuations around the Hartree evolution satisfy a central limit theorem. Interestingly, the variance of the limiting Gaussian distribution is determined by a time-dependent Bogoliubov transformation describing the dynamics of initial coherent states in a Fock space representation of the system.
\end{abstract}

\section{Introduction}
\label{sec:intro}
\setcounter{equation}{0}

A quantum mechanical system of $N$ bosons in three dimensions is described by a wave-function $\psi_N \in L^2_s (\bR^{3N})$, the subspace of $L^2 (\bR^{3N})$ consisting of functions which are symmetric w.r.t. permutations of the $N$ particles. Every physical observable of the system is associated with a self-adjoint operator $A$ on the Hilbert space $L^2 (\bR^{3N})$. The expectation of $A$ in the state described by the wave function $\psi_N$ is given by the $L^2$-inner product $\langle \psi_N , A \psi_N \rangle$. For example, the position of the $j$-th particle is associated with the multiplication operator $A = x_j$. The momentum of the $j$-th particle, on the other hand, is associated with the differential operator $A = - i\nabla_{x_j}$.

The time evolution of the quantum system is governed by the $N$-particle Schr\"odinger equation
\[ i\partial_t \psi_{N,t} = H_N \psi_{N,t} \]
where $\psi_{N,t} \in L_s^2 (\bR^{3N})$ indicates the wave-function of the system at time $t$. On the r.h.s., $H_N$ is a self-adjoint operator acting on $L^2 (\bR^{3N})$, known as the Hamilton operator, or the Hamiltonian, of the system. It typically has the form
\[ H = \sum_{j=1}^N -\Delta_{x_j} + \lambda \sum_{i<j} V (x_i - x_j) \]
where $\lambda \in \bR$ is a coupling constant, and $V$ is a two-body interaction among the particles. One could also introduce a term describing an external potential (for example, a trapping potential $V_{\text{ext}} (x) \to \infty$ as $|x| \to \infty$); the analysis below applies also in this case (under minor assumptions on $V_{\text{ext}}$, needed essentially to guarantee the self-adjointness of $H$).

We are interested in the mean-field regime, where particles undergo a large number of very weak collisions, so that the sum of the interactions experienced by each particle can effectively be approximated by an average, mean-field, potential. The mean-field regime is realized when $N \gg 1$ (many collisions), $|\lambda| \ll 1$ (weak interactions), so that $N \lambda =: \kappa$ is of order one (so that the total force acting on each particle is of order one). In order to study the mean-field regime, we analyze the dynamics generated by the Hamilton operator
\begin{equation}\label{eq:ham} H_N = \sum_{j=1}^N -\Delta_{x_j} + \frac{\kappa}{N} \sum_{i<j} V (x_i -x_j) \end{equation}
in the limit of large $N$. In particular, we are interested in the evolution of factorized initial data, given by
\[ \psi_N (x_1, \dots , x_N) = \prod_{j=1}^N \ph (x_j) \]
Although factorization is not preserved by the time-evolution, because of its mean-field character, it turns out that, for large $N$, factorization is approximately (and in an appropriate sense) preserved, i.e.
\begin{equation}\label{eq:conv} \psi_{N,t} (x_1, \dots , x_N) \simeq \prod_{j=1}^N \ph_t (x_j) \,. \end{equation}
A simple argument then shows that $\ph_t$ must evolve according to the effective nonlinear one-particle Hartree equation
\begin{equation}\label{eq:hartree} i\partial_t \ph_t = -\Delta \ph_t + \kappa \left( V * |\ph_t|^2 \right) \ph_t \,. \end{equation}
More precisely, it turns out that (\ref{eq:conv}) can be understood as convergence of the reduced density matrices associated with $\psi_{N,t}$. We define the density matrix associated with $\psi_{N,t}$ as the orthogonal projection $\gamma_{N,t} = |\psi_{N,t} \rangle \langle \psi_{N,t} |$ onto $\psi_{N,t}$. Then, for $k=1,\dots , N$, the $k$-particle reduced density matrix $\gamma^{(k)}_{N,t}$ associated with $\psi_{N,t}$ is defined as the partial trace of $\gamma_{N,t}$ over the last $(N-k)$ particles. In other words, $\gamma^{(k)}_{N,t}$ is defined as a non-negative, trace class operator on $L^2 (\bR^{3k})$ with integral kernel given by
\[ \begin{split} \gamma^{(k)}_{N,t} (x_1, &\dots , x_k ; x'_1, \dots , x'_k) \\ = \; & \int dx_{k+1} \dots dx_N \, \gamma_{N,t} (x_1, \dots , x_k, x_{k+1}, \dots , x_N ; x'_1, \dots, x'_k, x_{k+1}, \dots , x_N) \\ = \; &\int dx_{k+1} \dots dx_N \, \psi_{N,t} (x_1, \dots , x_k, x_{k+1}, \dots , x_N) \overline{\psi}_{N,t} ( x'_1, \dots, x'_k, x_{k+1}, \dots , x_N) \end{split} \]
For a $k$-particle observable $O^{(k)}$ (an observable depending non-trivially only on $k$ particles), we find
\[ \left\langle \psi_{N,t} , \left(O^{(k)} \otimes 1 \right) \psi_{N,t} \right\rangle = \tr\, \left(O^{(k)} \otimes 1 \right) \, \gamma_{N,t} = \tr\, O^{(k)} \gamma^{(k)}_{N,t} \]
It turns out that the language of the reduced density matrices is the correct language to understand the approximation (\ref{eq:conv}).
\begin{theorem}[Theorem 1.1 in \cite{CLS}] \label{thm:CLS}
Suppose that the potential $V$ satisfy the operator inequality $V(x)^2 \leq D (1-\Delta)$, for some constant $D>0$, and let $H_N$ be defined as in (\ref{eq:ham}). Let $\psi_{N} = \ph^{\otimes N}$, $\psi_{N,t} = e^{-iH_N t} \psi_N$. Then, for every $k\in \bN$ there exist constants $C,K >0$ such that
\begin{equation}\label{eq:conv2} \tr \left| \gamma^{(k)}_{N,t} \to |\ph_t\rangle \langle \ph_t|^{\otimes k} \right| \leq C \frac{e^{K |t|}}{N} \end{equation} for all $t \in \bR$ and all $N$ large enough. Here $\ph_t$ is the solution of the Hartree equation (\ref{eq:hartree}), with initial data $\ph_{t=0} = \ph$.
\end{theorem}
Remark that the operator inequality $V^2 \leq D (1-\Delta)$, which means that
\[ \int dx \, V^2 (x) |\ph (x)|^2 \leq D \| \ph \|_{H^1}^2 \]
for every $\ph \in L^2 (\bR^3)$, is satisfied by potentials with a Coulomb singularity $V(x) = \pm 1/|x|$.

Note that (\ref{eq:conv2}) implies that, for any $k$-particle observable $O^{(k)} \otimes 1^{(N-k)}$, we have
\[ \langle \psi_{N,t} , \left(O^{(k)} \otimes 1^{(N-k)} \right) \psi_{N,t} \rangle \to \langle \ph_t^{\otimes k} , O^{(k)} \ph_t^{\otimes k} \rangle \]
as $N \to \infty$, with rate of convergence of order $N^{-1}$. In other words (\ref{eq:conv2}) implies that, to compute the expectation of an arbitrary observable depending non-trivially on a fixed number of particles, in the limit of large $N$, we can replace the full solution of the $N$-particle Schr\"odinger equation by products of the solution of the Hartree equation~(\ref{eq:hartree}).

The first result in the direction of Theorem \ref{thm:CLS} was obtained in \cite{Sp} for bounded interaction potentials, studying the evolution of the reduced densities, as described by the so called BBGKY hierarchy (this approach does not give an estimate on the rate of convergence). The approach of \cite{Sp} was then extended to potential with singularities in \cite{EY,BGM,ES,ESY1,KSS,CP2}. More recently, it was also extended in \cite{ESY1,ESY2,ESY3} to the ultradilute limit of quantum mechanics, where particles interact through strong potentials with scattering length of the order $N^{-1}$; in this case, the many body dynamics is asymptotically approximated by the Gross-Pitaevskii equation (mathematical questions related with the Gross-Pitaevskii limit were also discussed in \cite{KM,CP1}). A different derivation of the Gross-Pitaevskii equation has been proposed in \cite{P}.

Another approach to the mean-field limit of quantum mechanics, based on ideas from \cite{He,GV} has been obtained in \cite{RS} and later extended in \cite{MS,CLS}. This approach is based on a representation of the system on the bosonic Fock-space and on the use of initial coherent states and it has the advantage of giving precise bounds on the rate of convergence towards the Hartree dynamics. Theorem \ref{thm:CLS}, as stated above, has been proven in \cite{CLS} making use of these techniques. The approach of \cite{RS} is also the basis for the analysis in the present paper; we will therefore describe it in details starting in Section~\ref{sec:fock}.
Alternative approaches and questions related with the mean-field limit of quantum dynamics have been discussed in \cite{ErS,KP,GMM,GMM2,C}.

Because of the bosonic symmetry, physically measurable observables on $L^2_s (\bR^{3N})$ are invariant w.r.t. permutations of the $N$ particles. A physically measurable one particle observable has the form $\cO = \sum_{j=1}^N O^{(j)}$, where $O^{(j)} =1 \otimes \dots \otimes 1 \otimes O \otimes 1 \otimes \dots \otimes 1$ acts non-trivially only on the $j$-th particle, for a self-adjoint operator $O$ on $L^2 (\bR^3)$. The kinetic energy $\sum_{j=1}^N -\Delta_{x_j}$ is an example of such an observable. The operator $\sum_{j=1}^N \chi_A (x_j)$, measuring the number of particles in a set $A \subset \bR^3$ is another example. In probabilistic terms, one can think of the observables $O^{(j)}$ as random variables, whose distribution is determined by the wave function $\psi_N \in L_s^2 (\bR^{3N})$ describing the state of the system (the distribution is determined through the spectral decomposition of the observables). If $\psi_N = \ph^{\otimes N}$, the $O^{(j)}$ are independent identically distributed random variables with a common distribution determined by the one-particle wave function $\ph \in L^2 (\bR^3)$ and, as a sum of independent and identically distributed random variables, $\cO = \sum_{j=1}^N O^{(j)}$ satisfies a law of large numbers and a central limit theorem. If we let the system evolve, $\psi_{N,t} = e^{-iH_N t} \ph^{\otimes N}$ does not factorize. Hence, at time $t \not = 0$, the variables $O^{(j)}$ are not independent (they are still identically distributed because of the permutation symmetry of $\psi_{N,t}$). Nevertheless, Theorem \ref{thm:CLS} implies that, asymptotically, $\psi_{N,t}$ can still be approximated in a certain sense by a factorized wave function. As it turns out, the result of Theorem \ref{thm:CLS} easily implies that, for any time $t \in \bR$, $\cO = \sum_{j=1}^N O^{(j)}$ still satisfies the law of large numbers (for a bounded self-adjoint operator $O$ on $L^2 (\bR^3)$). In fact, let $\wt{O} = O - \E_{\ph_t} O = O - \langle \ph_t, O \ph_t \rangle$, and let $\PP_{\psi_{N,t}}$ and $\E_{\psi_{N,t}}$ denote probabilities, respectively, expectations w.r.t. the distribution determined by the wave function $\psi_{N,t}$. Then, by Markov's inequality,
\[ \begin{split}  \PP_{\psi_{N,t}} \left( \left| \frac{1}{N} \sum_{j=1}^N \wt{O}^{(j)} \right| \geq \eps \right) &\leq \frac{1}{\eps^2 N^2} \E_{\psi_{N,t}} \left( \sum_{j=1}^N \wt{O}^{(j)} \right)^2 \\ & \leq \frac{1}{\eps^2} \left\langle \psi_{N,t}, \wt{O}^{(1)} \wt{O}^{(2)} \psi_{N,t} \right\rangle + \frac{1}{N\eps^2} \left\langle \psi_{N,t}, (\wt{O}^{(1)} )^2 \psi_{N,t} \right\rangle \\ & = \frac{1}{\eps^2} \tr \, \gamma^{(2)}_{N,t} (\wt{O} \otimes \wt{O}) + \frac{1}{N\eps^2} \tr \, \gamma^{(1)}_{N,t} \wt{O}^2
\end{split} \]
{F}rom Theorem \ref{thm:CLS}, we conclude that
\[
\lim\sup_{N\to \infty} \PP_{\psi_{N,t}} \left( \left| \frac{1}{N} \sum_{j=1}^N \wt{O}^{(j)} \right| \geq \eps \right) \leq \frac{1}{\eps^2} \tr \, |\ph_t \rangle \langle \ph_t|^2 \, \left( \wt{O} \otimes \wt{O} \right) = 0 \]

In the next sections, we show that, for all times $t\in \bR$, the one-particle observable $\cO$ also satisfies a central limit theorem, in the sense that, as $N \to \infty$
\[ \frac{1}{\sqrt{N}} \sum_{j=1}^N \left(O^{(j)} - \E_{\ph_t} O \right) \]
converges in distribution to a centered Gaussian random variable. The variance of this Gaussian variable is different from the one we would observe replacing the solution $\psi_{N,t}$ of the Schr\"odinger equation by the product $\ph_t^{\otimes N}$ (the correlations among the particles developed by the $N$-particle dynamics are sufficiently weak to still have a central limit theorem, but they are sufficiently strong to affect the variance of the limiting Gaussian fluctuations). As we will see the limiting variance is determined by a so called Bogoliubov transformation which describes the dynamics of fluctuations around the mean field Hartree dynamics in a Fock space representation of the system, which will be defined in the next section.

This is not the first time a central limit theorem is proven in a quantum setting. Quantum central limit theorems stating the convergence of polynomials in non-commuting self-adjoint observables to appropriate Gaussian distributions already appeared in the 70's, in the works \cite{CH,HL}. Since then, these results have been extended and applied to questions in equilibrium quantum statistical mechanics in \cite{GVV}, non-equilibrium quantum statistical mechanics in \cite{JPP}, quantum information theory in \cite{Ha}, and combinatorics in \cite{Ku}. The paper \cite{CE} is closer to the spirit of our manuscript, since it deals with quantum dynamics and with non-independent variables; however, in our case the correlations are generated by the time evolution (the particles are uncorrelated at time $t=0$), while in \cite{CE} they are already present in the initial state
(the Hamiltonian in this case is quadratic and does not generate correlations for factorized 
initial data).

\section{Fock space representation and fluctuation dynamics}
\label{sec:fock}
\setcounter{equation}{0}

To state the central limit theorem, we first have to study the fluctuations of the many-body evolution around the limiting Hartree dynamics. To this end, it is convenient to introduce a Fock-space representation of the system under consideration; this will allow us to study the evolution of so-called coherent states, for which the fluctuations can be written in a nice and compact form.

The bosonic Fock-space over $L^2 (\bR^3)$ is defined as the direct sum
\[ \cF = \bC \oplus \bigoplus_{n \geq 1} L^2_s (\bR^{3n} , dx_1, \dots dx_n) \]
where $L^2_s (\bR^{3n})$ denotes the subspace of $L^2 (\bR^{3n})$ consisting of functions symmetric with respect to permutation of the $n$ particles. Vectors in the Fock-space are sequences $\Psi = \{ \psi^{(n)} \}_{n \geq 0}$, where $\psi^{(n)} \in L^2_s (\bR^{3n})$ is an $n$-particle bosonic wave function. The idea behind the introduction of the Fock-space is that we want to study states with non fixed number of particles. Clearly, $\cF$  has the structure of a Hilbert space with the inner product
\[ \langle \Psi , \Phi \rangle = \overline{\psi^{(0)}} \phi^{(0)} + \sum_{n \geq 1} \langle \psi^{(n)} , \phi^{(n)} \rangle \, . \]
The vector $\Omega = \{ 1, 0, 0, \dots \} \in \cF$ is called the vacuum and describes a system with no particles. An important operator on $\cF$ is the number of particle operator $\cN$, which is defined by
\[ \cN \{ \psi^{(n)} \}_{n\geq 0} = \{ n \psi^{(n)} \}_{n\geq 0} \,.\]
The vacuum $\Omega$ is an eigenvector of $\cN$ with eigenvalue zero. More generally, vectors of the form $\{ 0, \dots , 0, \psi^{(m)}, 0, \dots \}$, having a fixed number of particles, are eigenvectors of $\cN$ (with eigenvalue $m$). The number of particle operator is an example of second quantization of a self-adjoint operator on the one-particle space $L^2 (\bR^3)$. In general, if $O$ denotes a self-adjoint operator on $L^2 (\bR^3)$, we define its second quantization $d\Gamma (O)$ as the self-adjoint operator on $\cF$ defined by
\[ \left(d\Gamma (O) \psi \right)^{(n)} = \sum_{j=1}^n O^{(j)} \psi^{(n)} \]
where $O^{(j)} = 1 \otimes \dots \otimes O \otimes \dots \otimes 1$ acts as $O$ on the $j$-th particle and as the identity on the other $(n-1)$ particles. With this notation $\cN = d\Gamma (1)$.

On $\cF$, we define the Hamilton operator $\cH_N$ by $\cH_N \{ \psi^{(n)} \}_{n \geq 1} = \{ \cH_N^{(n)} \psi^{(n)} \}_{n \geq 1}$, with
\begin{equation}\label{eq:fock-ham1} \cH_N^{(n)} = \sum_{j=1}^n -\Delta_{x_j} + \frac{1}{N} \sum_{i<j}^n V (x_i - x_j) \,.\end{equation}
By definition, $\cH_N$ leaves each $n$-particle sector $\cF_n$ (defined as the eigenspace of $\cN$ associated with the eigenvalue $n$) invariant. Moreover, on the $N$-particle sector, $\cH_N$ agrees with the mean field Hamiltonian $H_N$ defined in (\ref{eq:ham}). In particular, if we consider the Fock-space evolution of an initial vector with exactly $N$ particles, we find
\[ e^{-i\cH_N t} \{ 0, \dots , 0, \psi_N, 0, \dots \} = \{ 0, \dots, 0, e^{-i H_N t} \psi_N , 0, \dots \} \]
exactly as in Section \ref{sec:intro}. The advantage of working in the Fock space is that we have more freedom in the choice of the initial data. We will use this freedom by considering a class of initial data, known as coherent states, with non-fixed number of particles.

It is very useful to introduce, on the Fock space $\cF$, creation and annihilation operators. For $f \in L^2 (\bR^3)$, we define the creation operator $a^* (f)$ and the annihilation operator $a(f)$ by setting
\begin{equation}\label{eq:aa^*} \begin{split} \left( a^* (f) \psi \right)^{(n)}
(x_1, \dots ,x_n) &= \frac{1}{\sqrt{n}} \sum_{j=1}^n f (x_j)
\psi^{(n-1)} (x_1 , \dots, x_{j-1}, x_{j+1},\dots, x_n) \\ \left( a (f) \psi
\right)^{(n)} (x_1, \dots ,x_n) &= \sqrt{n+1} \int \rd x \;
\overline{f (x)} \psi^{(n+1)} (x, x_1 , \dots , x_n) \, . \end{split} \end{equation}
The operators $a^* (f)$ and $a(f)$ are densily defined and closed. It is easy to check that, as the notation suggests, $a^* (f)$ is the adjoint of $a(f)$. Creation and annihilation operators satisfy the canonical commutation relations
\begin{equation}\label{eq:CCR} \left[ a(f) , a^* (g) \right] = \langle f,g \rangle \qquad \text{and } \quad \left[ a (f) , a (g) \right] = \left[ a^* (f) , a^* (g) \right] = 0 \end{equation}
for any $f,g \in L^2 (\bR^3)$ (here $\langle f,g\rangle$ denotes the $L^2$-inner product). Physically, the operator $a^* (f)$ creates a particle with wave function $f$, while $a (f)$ annihilates it. As a consequence, a state with $N$ particles all in the one-particle state $\ph$, can be written as
\[ \{ 0, \dots, 0, \ph^{\otimes N} , 0, \dots \} = \frac{(a^* (\ph))^N}{\sqrt{N!}} \Omega \,. \]
It is also useful to introduce operator-valued distributions $a_x, a^*_x$ defined so that
\begin{equation}\label{eq:aa^*-dis} a(f) = \int dx \overline{f (x)} a_x, \qquad \text{and } \quad a^* (f) = \int dx \, f(x) \, a^*_x \, .\end{equation}
With this notation, $a_x^* a_x$ gives the density of particles close to $x\in \bR^3$. The number of particles operator can formally be written as
\[ \cN = \int dx \, a_x^* a_x \, .\]
More generally, the second quantization of a self-adjoint one-particle operator $O$ on $L^2 (\bR^3)$ with integral kernel $O(x,y)$ can be expressed as
\[ d\Gamma (O) = \int dx dy \, O(x,y) a_x^* a_y \, . \]
Similarly, the Hamilton operator $\cH_N$ can be formally expressed in terms of operator-valued distributions as
\begin{equation}\label{eq:fock-ham2} \cH_N = \int dx \, \nabla_x a^*_x \, \nabla_x a_x  +  \frac{1}{2N} \int dx dy \, V (x-y) a_x^* a_y^* a_y a_x \,. \end{equation}
The fact that every term in the Hamiltonian contains the same number of creation and annihilation operators means that $\cH_N$ commutes with the number of particles or, equivalently, that the number of particles is preserved by the time-evolution enerated by $\cH_N$.

For later use, we observe that the creation and annihilation operators are not bounded; however, they can be bounded by the square root of the number of particles operator, in the sense that
\begin{equation}\label{eq:bd-aa*}
\begin{split}\| a(f) \psi \| &\leq \| f \| \, \| \cN^{1/2} \psi \| \qquad \text{and } \qquad  \| a^* (f) \psi \| \leq \| f \| \, \| (\cN +1)^{1/2} \psi \|
\end{split}\end{equation}
for every $\psi \in \cF$, $f \in L^2 (\bR^3)$ (here $\| f \|$ indicates the $L^2$-norm of $f$).

As mentioned above, we are going to study the evolution of initial coherent states.
For arbitrary $\ph \in L^2 (\bR^3)$, we define the Weyl operator
\begin{equation}\label{eq:weyl} W (\ph) = e^{a^* (\ph) - a(\ph)} \,.\end{equation}
The coherent state with wave function $\ph$ is then defined as $W(\ph) \Omega$. The Weyl operator $W(\ph)$ is a unitary operator; therefore the coherent state $W(\ph) \Omega$ always has norm one. Moreover, since
\[ W (\ph) \Omega = e^{-\| \ph \|^2/2} \sum_{j=0} \frac{a^* (\ph)^j}{j!} \Omega = e^{-\| \ph \|^2/2} \{ 1, \ph, \frac{\ph^{\otimes 2}}{\sqrt{2!}}, \dots, \frac{\ph^{\otimes j}}{\sqrt{j!}}, \dots \} \, , \]
the coherent state $W(\ph) \Omega$ does not have a fixed number of particles. One can nevertheless compute the expectation of the number of particles in the state $W (\ph) \Omega$; it is given by
\[ \langle W (\ph) \Omega, \cN W (\ph) \Omega \rangle = \| \ph \|^2 \, . \]
More precisely, it turns out that the number of particle in the coherent state $W(\ph) \Omega$ is a Poisson random variable with expectation and variance $\| \ph \|^2$. The main reason why coherent states have nice algebraic properties (which will be used later on in the analysis of their evolution) is the fact that they are eigenvectors of all annihilation operators. Indeed
\[ a (f) W (\ph) \Omega = (f, \ph) W(\ph) \Omega \]
for every $f,\ph \in L^2 (\bR^3)$. This is a simple consequence of the fact that Weyl operators generate shifts of creation and annihilation operators, in the sense that
\begin{equation}\label{eq:shift} W^* (\ph) a (f) W(\ph) = a (f) + (f, \ph), \qquad \text{and } \quad W^* (\ph) a^* (f) W(\ph) = a^* (f) + (\ph , f) \, . \end{equation}
Instead of studying the time evolution of initially factorized states, of the form $\psi_N = a^* (\ph)^N \Omega / \sqrt{N!}$, it is convenient to analyze the dynamics of coherent states of the form $W (\sqrt{N} \ph) \Omega$, for $\ph \in L^2 (\bR^3)$ with $\| \ph \| =1$ (at the end, it is possible to express factorized states as linear combinations of coherent states, and therefore to establish properties of their evolution as well).
The expected number of particles in the state $W(\sqrt{N} \ph) \Omega$ is equal to $N$. For this reason, we may expect the evolution of $W (\sqrt{N} \ph) \Omega$ to have the same mean-field properties as the one of $\psi_N$. In other words, we may expect that the evolution of $W(\sqrt{N} \ph) \Omega$ is still approximated by the coherent state $W(\sqrt{N} \ph_t) \Omega$, where $\ph_t$ is the solution of the Hartree equation (\ref{eq:hartree}), with initial data $\ph_{t=0} = \ph$. To study the fluctuation around the mean-field approximation $W(\sqrt{N} \ph_t) \Omega$, we introduce the two-parameter group of unitary transformations
\begin{equation}\label{eq:cU}
\cU (t;s) = e^{-i \omega(t;s)} \, W^* (\sqrt{N} \ph_t) e^{-i \cH_N (t-s)} W (\sqrt{N} \ph_s) \end{equation}
for all $t,s \in \bR$, with the phase factor
\[ \omega(t;s) = \frac{N}{2} \int_s^t d\tau \int dx (V*|\ph_{\tau}|^2 ) (x) |\ph_{\tau} (x)|^2. \]
It turns out that $\cU (s;s) = 1$ and that $\cU(t;s)$ is strongly differentiable on the domain $\cD (\cK + \cN) \subset \cF$ (at least under the assumptions on $\ph$ and $V$ given in Theorem \ref{thm:CLS}). Here $\cK$ is the kinetic energy operator
\begin{equation}\label{eq:cK0} \cK = d\Gamma (-\Delta) = \int dx \, \nabla_x a_x^* \, \nabla_x a_x \end{equation}
The derivative of $\cU (t;s)$ is then given by
\begin{equation}\label{eq:cU2} i\partial_t \cU (t;s) = \cL(t) \cU (t;s), \end{equation}
where we introduced the generator
\begin{equation}\label{eq:L} \begin{split} \cL (t) = \; & \int dx \, \nabla_x a^*_x \nabla_x a_x  + \int dx \, (V*|\ph_t|^2) (x) a_x^* a_x + \int dx dy \, V(x-y) \, \ph_t (x) \overline{\ph}_t (y) a_x^* a_y \\ & + \frac{1}{2} \int dx dy \, V(x-y) \, \left( \ph_t (x) \ph_t (y) a^*_x a^*_y + \overline{\ph}_t (x) \overline{\ph}_t (y) a_x a_y \right) \\ &+\frac{1}{ \sqrt{N}} \int dx dy V (x-y) \, a_x^* \left( a_y^* \ph_t (y) + a_y \overline{\ph}_t (y) \right) a_x \\ &+\frac{1}{2N} \int dx dy \, V( x-y) \, a_x^* a_y^* a_y a_x \,.\end{split} \end{equation}
The phase factor $\omega (t;s)$ in (\ref{eq:cU}) really does not play an important role (it cancels from expressions of the form $\cU^* (t;s) A \cU (t;s)$, for any operators $A$ on $\cF$). To understand the effect of the unitary evolution $\cU (t;s)$, we observe that, for a one-particle operator $O$ on $L^2 (\bR^3)$ with integral kernel $O(x,y)$, we have
\[ \begin{split}
\frac{1}{N} \Big\langle e^{-i \cH_N t} &W(\sqrt{N} \ph) \Omega , \, d\Gamma (O) \, e^{-i \cH_N t} W(\sqrt{N} \ph) \Omega \Big\rangle \\ &= \frac{1}{N} \int dx dy O(x,y) \, \left\langle e^{-i \cH_N t} W(\sqrt{N} \ph) \Omega , \, a_x^* a_y \, e^{-i \cH_N t} W(\sqrt{N} \ph) \Omega \right\rangle \\ &= \frac{1}{N} \int dx dy O(x,y) \, \left\langle \cU (t;0) \Omega , (a_x^* + \sqrt{N} \overline{\ph}_t (x))(a_y + \sqrt{N} \ph_t (y)) \, \cU (t;0) \Omega \right\rangle \\ &= \langle \ph_t , O \ph_t \rangle + \frac{1}{N} \langle \cU (t;0) \Omega, d\Gamma (O) \cU (t;0) \Omega \rangle + \frac{1}{\sqrt{N}} \langle \cU (t;0) \Omega, \phi (O\ph_t) \cU (t;0) \Omega \rangle \end{split} \]
Hence, $\cU (t;s)$ describes the evolution of the fluctuations around the mean-field dynamics. As $N \to \infty$, the last two terms on the r.h.s. of (\ref{eq:L}) vanish. We can therefore expect the fluctuation dynamics $\cU (t;s)$ to converge in this limit towards a limiting evolution with time-dependent generator
\begin{equation}\label{eq:Linfty}\begin{split}
\cL_{\infty} (t) = \; & \int dx \, \nabla_x a^*_x \nabla_x a_x + \int dx \, (V*|\ph_t|^2) (x) a_x^* a_x + \int dx dy \, V(x-y) \, \ph_t (x) \overline{\ph}_t (y) a_x^* a_y \\ & + \frac{1}{2} \int dx dy \, V(x-y) \, \left( \ph_t (x) \ph_t (y) a^*_x a^*_y + \overline{\ph}_t (x) \overline{\ph}_t (y) a_x a_y \right) \end{split} \end{equation}
The existence of such a limiting dynamics was established in \cite{GV}; for completeness, we reproduce the relevant statement in the next proposition.
\begin{proposition}[Prop. 4.1 in \cite{GV}]\label{prop:GV}
Suppose $V \in L^\infty (\bR^3) + L^2 (\bR^3)$, and that the map $t \to \ph_t$ is in $C (\bR, L^2 (\bR^3) \cap L^4 (\bR^3))$ (both conditions are easily checked under the assumptions of Theorem \ref{thm:CLT} below). Then there exists a unique two-parameter group of unitary transformations $\cU_\infty (t;s)$ with $\cU_\infty (s;s) = 1$ for all $s \in \bR$, and such that $\cU_\infty (t;s)$ is strongly differentiable on the domain $\cD (\cK + \cN)$ ($\cK$ is the kinetic energy operator from (\ref{eq:cK0})) with
\begin{equation}\label{eq:Uinfty} i\partial_t \cU_\infty (t;s) = \cL_\infty (t) \cU_\infty (t;s) \,. \end{equation}
\end{proposition}

It turns out that the limiting dynamics $\cU_\infty (t;s)$ can be described as a so called Bogoliubov transformation. This is a consequence of the fact that the time-dependent generator $\cL_\infty (t)$ is quadratic in the creation and annihilation operators. For $f,g \in L^2 (\bR^3)$, we introduce the notation
\[ A(f,g) = a(f) + a^* (\overline{g}) \]
We observe that \begin{equation}\label{eq:bog11} (A(f,g))^* = A(\overline{g}, \overline{f}) = A \left( \left( \begin{array}{ll} 0 & J \\ J & 0 \end{array} \right) (f,g) \right)\end{equation}
where $J:L^2 (\bR^3) \to L^2 (\bR^3)$ is the antilinear operator defined by $Jf = \overline{f}$. {F}rom (\ref{eq:CCR}), we find that the operators $A(f,g)$ satisfy the commutation relations
\begin{equation}\label{eq:bog22} \left[ A(f_1, g_1) , A^* (f_2, g_2) \right] = \left\langle (f_1, g_1) , S (f_2, g_2) \right\rangle_{L^2 \oplus L^2} \qquad \text{with } S = \left( \begin{array}{ll} 1 & 0 \\ 0 & -1 \end{array} \right) \end{equation}

\begin{theorem}\label{thm:bog}
Assume $V^2 \leq D (1-\Delta)$ and $\ph \in H^1 (\bR^3)$.
Then, for every $t,s\in \bR$ there exists a bounded linear map \[ \Theta (t;s): L^2 (\bR^3) \oplus L^2 (\bR^3) \to L^2 (\bR^3) \oplus L^2 (\bR^3) \] such that
\[ \cU_\infty^* (t;s) A(f,g) \cU_\infty (t;s) = A (\Theta (t;s) (f,g)) \]
for all $f,g \in L^2 (\bR^3)$. The map $\Theta (t;s)$ is so that, for every $t,s \in \bR$,
\begin{equation}\label{eq:bog1}
\Theta (t;s) \left( \begin{array}{ll} 0 & J \\ J & 0 \end{array} \right) = \left( \begin{array}{ll} 0 & J \\ J & 0 \end{array} \right) \Theta (t;s) \end{equation}
and
\begin{equation}\label{eq:bog2} S = \Theta (t;s)^* \, S \, \Theta (t;s) \end{equation}
Moreover, $\Theta (t;s)$ is such that
\begin{equation}\label{eq:Ttph} \Theta (t;s) (\ph_t , \overline{\ph}_t) = (\ph_s, \overline{\ph}_s)\, . \end{equation}

A map $\Theta (t;s)$ with the properties (\ref{eq:bog1}) and (\ref{eq:bog2}) is known as a (bosonic) Bogoliubov transformation. Note that (\ref{eq:bog1}) and (\ref{eq:bog2}) simply state that (\ref{eq:bog11}) and, respectively, (\ref{eq:bog22}) are preserved under $\Theta (t;s)$.
Like every Bogoliubov transformation, $\Theta (t;s)$ can be written as
\begin{equation}\label{eq:thetat} \Theta (t;s) = \left( \begin{array}{ll} U (t;s) & J V (t;s) J \\ V (t;s) & J U (t;s) J \end{array} \right) \end{equation}
for bounded linear maps $U (t;s), V (t;s) : L^2 (\bR^3) \to L^2 (\bR^3)$ satisfying \[ U^* (t;s) U (t;s) - V^* (t;s) V (t;s) = 1 \qquad \text{and } \quad U^* (t;s) J V (t;s) J = V^* (t;s) J U (t;s) J.\]
The maps $U(t;s), V(t;s)$ are not only bounded from $L^2 (\bR^3)$ into $L^2 (\bR^3)$. Instead they even extend to bounded maps from $H^1 (\bR^3)$ into $H^1 (\bR^3)$ (and $\Theta (t;s)$ extend therefore to a bounded maps from $H^1 (\bR^3) \oplus H^1 (\bR^3)$ into itself).
\end{theorem}

{\it Remark.} The Bogoliubov transformation $\Theta (t;s)$ is a two-parameter group of symplectic transformations. Formally, they satisfy the equation
\begin{equation}\label{eq:dtTheta} i\partial_t \Theta (t;s) = \Theta (t;s) \cA (t) \end{equation}
with generator
\[ \cA (t) = \left( \begin{array}{ll} D_t & -JB_t J \\ B_t & -JD_t J \end{array} \right) \]
where we defined the linear operators
\[ \begin{split} D_t f &= -\Delta f + (V*|\ph_t|^2) f + (V* \overline{\ph}_t f) \ph_t \\ B_t f &= (V* \overline{\ph}_t f) \overline{\ph}_t \end{split} \]
acting on $L^2 (\bR^3)$. We observe that $D_t^* = D_t$, $B_t^* = J B_t J$. and therefore $\cA_t^* = S \cA_t S$ ($\cA (t)$ is self-adjoint w.r.t. the degenerate product defined by $S$; hence $\Theta (t;s)$ is unitary w.r.t. this inner product, which means that $\Theta (t;s)$ is symplectic). Note also that the condition $V^2 \leq D (1-\Delta)$ implies that $B_t$ is a Hilbert-Schmidt operator for all $t \in \bR$. This is an important fact to make the Bogoliubov generated by $\cA (t)$ implementable. Since the existence of a generator $\cA (t)$ for $\Theta (t;s)$ is not needed for our analysis, we do not try to make (\ref{eq:dtTheta}) precise.

We will prove Theorem \ref{thm:bog} in Section \ref{sec:bds}, where we establish several properties of the limiting evolution $\cU_\infty )t;s)$.

\section{The Central Limit Theorem}
\setcounter{equation}{0}
\label{sec:CLT}

We are now ready to state our main result.
\begin{theorem}\label{thm:CLT}
Suppose $V$ is a measurable function such that $V^2 (x) \leq D (1-\Delta)$, for some $D >0$. For $\ph \in H^2 (\bR^3)$ with $\| \ph \| =1$, let $\psi_N = \ph^{\otimes N}$ and $\psi_{N,t} = e^{-i H_N t} \psi_N$. Let $O$ be a bounded self adjoint operator on $L^2 (\bR^3)$ with $\| \nabla O (1-\Delta)^{-1/2} \| < \infty$. For arbitrary fixed $t \in \bR$, let
\begin{equation}\label{eq:obser} \cO_t : = \frac{1}{\sqrt{N}} \sum_{j=1}^N \left(O^{(j)} - \E_{\ph_t} O \right) = \frac{1}{\sqrt{N}} \sum_{j=1}^N \left( O^{(j)} - \langle \ph_t, O \ph_t \rangle \right) \end{equation}
where $O^{(j)} = 1^{(j-1)} \otimes O \otimes 1^{(N-j)}$ denotes the operator $O$ acting on the $j$-th particle. Then, as $N \to \infty$, the random variable $\cO_t$ (whose law is determined by the $N$-particle wave function $\psi_{N,t}$) converges in distribution to a centered Gaussian random variable with variance \begin{equation}\label{eq:var}\begin{split} \sigma_t^2
&=  \frac{1}{2} \left[ \left\langle \Theta (t;0) \left(O\ph_t , J O \ph_t \right) , \Theta (t;0) \left(O \ph_t , J O\ph_t \right) \right\rangle - \left| \left\langle \Theta (t;0) \left(O\ph_t , J O \ph_t \right) , \frac{1}{\sqrt{2}} \, \left(\ph, \overline{\ph} \right) \right\rangle \right|^2 \right] \\ &= \| U (t;0) O \ph_t + J V (t;0) O \ph_t \|^2 - |\langle \ph , U (t;0) O \ph_t + J V (t;0) O \ph_t \rangle|^2 \geq 0 \end{split} \end{equation}
where the bounded linear operators $U (t;0), V (t;0) : L^2 (\bR^3) \to L^2 (\bR^3)$ and $\Theta (t;0) : L^2 (\bR^3) \oplus L^2 (\bR^3) \to L^2 (\bR^3) \oplus L^2 (\bR^3)$ are defined in Theorem \ref{thm:bog}.
\end{theorem}

Notice that, for $t=0$, since $\Theta (0;0) = 1$ (and thus $U (0;0) = 1$, $V (0;0) = 0$), the variance is given by
\[ \sigma^2_{t=0}  =  \| O \ph \|^2 - \langle \ph , O \ph \rangle^2 = \E_{\ph} O^2 - (\E_{\ph} O )^2 \]
which is consistent with the fact that, at time $t=0$, $\psi_{N,t=0} = \ph^{\otimes N}$ and the variables $O^{(j)}$ are therefore independent identically distributed random variables, with a common distribution determined (through the spectral theorem) by $\ph$.

Instead of the variable $\cO_t$, we could also consider the variable
\[ \cO'_t = \frac{1}{\sqrt{N}} \sum_{j=1}^N \left( O^{(j)} - \langle \psi_{N,t} , O^{(\ell)} \psi_{N,t} \rangle \right) = \frac{1}{\sqrt{N}} \sum_{j=1}^N \left( O^{(j)} - \tr \, O \, \gamma_{N,t}^{(1)} \right) \]
which is centered for all $N$ (while $\cO_t$ is only centered in the limit $N \to \infty$). Observe that
\[ \cO'_t = \cO_t + \delta_N, \qquad \text{with } \quad \delta_N = \sqrt{N} \tr \, O \left( |\ph_t \rangle \langle \ph_t| - \gamma^{(1)}_{N,t} \right)\]
It follows from Theorem \ref{thm:CLS} that, for every fixed time,  $|\delta_N| \lesssim N^{-1/2}$. For this reason, like $\cO_t$, also the random variable $\cO'_t$ converges in distribution, as $N \to \infty$, to a Gaussian random variable with variance $\sigma_t^2$ as given in (\ref{eq:var}).

We also remark that the assumptions $\ph \in H^2 (\bR^3)$ and $\| \nabla O (1-\Delta)^{-1/2} \| < \infty$ are needed to control the singularity of the potential $V$. If instead of $V^2 (x) \leq D (1-\Delta)$ we just assumed that the operator norm of $V$ is bounded (i.e. that the $L^\infty$ norm of the function $V$ is finite), the result would hold for arbitrary $\ph \in H^1 (\bR^3)$ and for arbitrary compact operator $O$ and its proof would be substantially simpler (the analytic part of the proof would be easier, the combinatorial part would be unchanged). Since however we believe the Coulomb potential to be physically interesting, we preferred to avoid the assumption of bounded potential.

For example, Theorem \ref{thm:CLT} applies when $O = \chi_A (x)$, where $\chi_A$ is a smoothed out version of the characteristic function of a set $A \subset \bR^3$ (one needs the smoothed out version, because of the condition $\| \nabla \chi_A (1-\Delta)^{-1} \| < \infty$). In this case, the operator $\sum_{j=1}^N \chi_A (x_j)$ measures the number of particles in the set $A$. The law of large numbers implies that $\sum_{j=1}^N \chi_A (x_j)$ converges almost surely to 
\begin{equation}\label{eq:aver-ex} N \langle \ph_t, \chi_A (x) \ph_t \rangle = N \int dx \chi_A (x) |\ph_t (x)|^2 \end{equation}
as $N \to \infty$. Theorem \ref{thm:CLT} implies that the typical fluctuations around (\ref{eq:aver-ex}) are of the order $N^{1/2}$ and that, on this scale, they approach a Gaussian law.

In order to prove Theorem \ref{thm:CLT}, we show that the moments of the random variable $\cO_t$, with distribution determined by the wave function $\psi_{N,t}$ converge, as $N \to \infty$, to the moments of a centered Gaussian random variable with variance $\sigma_t^2$.
\begin{lemma}\label{lm:mom}
Under the assumptions of Theorem \ref{thm:CLT}, we have, for any fixed $\ell \in \bN$,
\begin{equation}\label{eq:mom} \lim_{N \to \infty} \E_{\psi_{N,t}} \cO_t^\ell = \lim_{N \to \infty} \left\langle \psi_{N,t} , \cO_t^{\ell} \, \psi_{N,t} \right\rangle = \left\{ \begin{array}{ll} 0 \quad &\text{if $\ell$ is odd} \\
\frac{\ell!}{2^{\ell/2} (\ell/2)!} \sigma_t^{2\ell} \quad &\text{if $\ell$ is even} \end{array} \right. \end{equation}
\end{lemma}

\section{Key bounds on fluctuation dynamics}
\label{sec:bds}
\setcounter{equation}{0}

To show Lemma \ref{lm:mom}, we need some crucial estimate on the fluctuation dynamics $\cU (t;s)$ defined in (\ref{eq:cU}) and on its formal limit $\cU_\infty (t;s)$ defined in (\ref{eq:Uinfty}). In particular, we need to control the growth of the expectation of the number of particles operator $\cN$, of its high powers $\cN^j$, and of the kinetic energy operator
\[ \cK = \int dx \nabla_x a^*_x \nabla_x a_x \]

We assume that the potential $V$ entering the definition of $\cU(t;s)$ and of $\cU_\infty (t;s)$ satisfies the operator inequality \begin{equation}\label{eq:Vass}
V^2 (x) \leq D (1-\Delta) \end{equation}
for some $D >0$. Moreover, $\ph_t$ is the solution of the Hartree equation (\ref{eq:hartree}) corresponding to an initial data $\ph \in H^2 (\bR^3)$, with $\| \ph \| = 1$. Note here that the bounds on the growth of the number of particles only depend on the $H^1$ norm of the initial data $\ph$, and therefore can be extended to all $\ph \in H^1 (\bR^3)$. The bounds on the growth of the kinetic energy, on the other hand, depends on $\| \dot{\ph}_t \|$. {F}rom the Hartree equation (\ref{eq:hartree}), we observe that
\begin{equation}\label{eq:dotph} \| \dot{\ph}_t \| \leq \| \Delta \ph_t \| + \| (V*|\ph_t|^2) \ph_t \| \leq C \| \ph_t \|_{H^2} \leq C \| \ph \|_{H^2} \end{equation}
uniformly in $t$ (the constant $C$ depends only on the constant $D$ in (\ref{eq:Vass})). Here, we used the propagation of regularity for solutions of the Hartree equation (\ref{eq:hartree}); for any $s \geq 1$, there exists a constant $C$ such that \begin{equation}\label{eq:phHs}
\sup_{t \in \bR} \| \ph_t \|_{H^s} \leq C \| \ph \|_{H^s}.
\end{equation}
We will use this fact in several occasions.

In the next proposition, we control the growth of powers of the number of particles, of the kinetic energy and of its square with respect to the limiting dynamics $\cU_\infty (t;s)$ defined in (\ref{eq:Uinfty}).
\begin{proposition}\label{prop:2}
For every $j \in \bN$, there exist constants $C, K >0$ (depending on $D$, $\| \ph \|_{H^1}$, and $j$) such that
\begin{equation}\label{eq:Ninfty} \left\langle \cU_\infty (t;s) \psi , \cN^j \, \cU_\infty (t;s) \psi \right\rangle \leq C e^{K |t-s|}
\left\langle \psi, (\cN +1)^{j} \, \psi \right\rangle  \end{equation}
for all $\psi \in \cF$, and all $t,s \in \bR$. Moreover, there are also constants $C,K >0$ (depending on $D$, $\| \ph \|_{H^1}$) and $C', K'$ (depending on $D$, and $\| \ph \|_{H^2}$) such that
\begin{equation}\label{eq:Kinfty} \left\langle \cU_\infty (t;s) \psi , \cK \, \cU_\infty (t;s) \psi \right\rangle \leq C e^{K |t-s|}
\left\langle \psi, (\cK + \cN +1) \, \psi \right\rangle \, \end{equation}
and
\begin{equation}\label{eq:K2infty}
\left\langle \cU_\infty (t;s) \psi , \cK^2 \, \cU_\infty (t;s) \psi \right\rangle \leq C' e^{K' |t-s|}
\left\langle \psi, (\cK^2 + \cN^2 +1) \, \psi \right\rangle \, \end{equation}
for all $\psi \in \cF$, and all $t,s \in \bR$.
\end{proposition}

\begin{proof}
Eq. (\ref{eq:Ninfty}) is taken from \cite{CLS}[Prop. 5.1].
To prove (\ref{eq:Kinfty}) we combine Lemma \ref{lm:KHN} below, where we prove the bound $\cK \leq \cL_\infty (t) + C (\cN +1)$ with the estimate
\[ \left| \langle \cU_\infty (t;s) \psi, \cL_\infty (t) \cU_\infty (t;s) \psi \rangle \right| \leq C e^{K|t-s|} \langle \psi, (\cL_\infty (s) + \cN +1 ) \psi \rangle
\] on the growth of the expectation of $\cL_\infty (t)$, proven in \cite{CLS}[Lemma 6.2].
%Notice that, in contrast with the definitions of $\cU_\infty (t;s)$ and %$\cL_\infty (t)$, in \cite{CLS} the evolution $\cU_2 (t;s)$ and its %generator $\cL_2 (t)$ are defined in terms of the regularized potential %and of the solution of the regularized Hartree equation. Nevertheless, %glancing through the proof of Lemma 6.1 and Lemma 6.2 it is immediately %clear that the cutoff does not play any role there, and that the proof %only uses the bound $V^2 \leq D (1-\Delta)$, and a bound on the $H^1$-norm %of $\ph_t$.
To prove (\ref{eq:K2infty}), we write
\[ \cK = \cL_\infty (t) - A_{1,t} - A_{2,t} - A_{3,t} \]
where
\[ \begin{split} A_{1,t} = \; & \int dx \, (V*|\ph_t|^2) (x) \, a_x^* a_x \\ A_{2,t} = \; & \int dx dy V(x-y) \ph_t (x) \overline{\ph}_t (y) a_x^* a_y \end{split} \]
and $A_{3,t} = B_t +B_t^*$, with
\[ B = \int dx dy \, V(x-y) \, \ph_t (x) \ph_t (y) \, a_x^* a_y^* \]
Hence, we have
\begin{equation}\label{eq:K2A} \cK^2 \lesssim \cL_\infty (t)^2 + A_{1,t}^2 + A_{2,t}^2 + A_{3,t}^2 \end{equation}
Observe that, because of the assumption $V^2 (x) \leq D (1-\Delta)$, and because of (\ref{eq:phHs}),
\begin{equation}\label{eq:A1t} |A_{1,t}| \leq \sup_{x} \left| V * |\ph_t|^2 \right|  \, \cN \leq C \cN \end{equation}
for a constant $C >0$ depending only on $D$ and on $\| \ph \|_{H^1}$. Similarly
\begin{equation}\label{eq:A2t} \begin{split} \left\langle \psi, A_{2,t} \psi \right\rangle &\leq \int dx dy \,|V(x-y)| |\ph_t (x)| |\ph_t (y)| \, \| a_x \psi \| \, \| a_y \psi \| \\ &\leq \int dx dy \, |V(x-y)|^2 |\ph_t (x)|^2 \| a_y \psi \|^2 + \int dx dy \, |\ph_t (y)|^2 \| a_x \psi \|^2
\\ &\lesssim  (D \| \ph_t \|_{H^1} + 1) \, \langle \psi , \cN \psi \rangle \,. \end{split} \end{equation}
Combining the last bound with the corresponding lower bound (and with (\ref{eq:phHs})), we find $|A_{2,t}| \leq C \cN$. Since $A_{1,t}$ and $A_{2,t}$ commute with $\cN$, it follows that
\begin{equation}\label{eq:A12t} A_{1,t}^2 \leq C \cN^2 \qquad \text{and } \quad A_{2,t}^2 \leq C \cN^2 \,. \end{equation}
We still have to consider the term $A_{3,t}$. To this end, we notice that
\begin{equation}\label{eq:A3t} A_{3,t}^2 \lesssim B_t B_t^* + B_t^* B_t \end{equation}
where \begin{equation}\label{eq:BB*} \begin{split} B_t B_t^* &= \int dx dy dz dw V(x-y) V(z-w) \ph_t (x) \ph_t (y) \overline{\ph}_t (z) \overline{\ph}_t (w) a_x^* a_y^* a_z a_w \\ & \leq  \int dx dy dz dw \, V^2 (x-y) |\ph_t (x)|^2 |\ph_t (y)|^2 \, a_z^* a_w^* a_w a_z \\ & \lesssim D \| \ph_t \|_{H^1}^2 \cN^2 \leq C \cN^2 \end{split}\end{equation}
and
\begin{equation}\label{eq:B*B} \begin{split}
B_t^* B_t = \; &B_t B_t^* + 4 \int dx dy dz \, V(x-y) V(y-z) |\ph_t (y)|^2 \ph_t (x) \overline{\ph}_t (z) a_x^* a_z \\ &+ 2 \int dx dy\, V^2 (x-y) |\ph_t (x)|^2 |\ph_t (y)|^2 \\  \leq \; &C (\cN+1)^2
\end{split} \end{equation}
for a constant $C$ depending only on $D$ and on the norm $\| \ph \|_{H^1}$. {F}rom (\ref{eq:K2A}), (\ref{eq:A12t}), (\ref{eq:BB*}), (\ref{eq:B*B}), we conclude that
\begin{equation}\label{eq:K2L2} \cK^2 \leq \cL_{\infty}^2 (t) + C (\cN+1)^2 \end{equation}
Similarly, one can also write $\cL_\infty (t) = \cK + A_{1,t} + A_{2,t} + A_{3,t}$, and prove that
\begin{equation}\label{eq:L2K2} \cL_\infty^2 (t) \leq \cK^2 + C (\cN+1)^2 \end{equation}

Next, we need to control the growth of the expectation of $\cL_\infty^2 (t)$ w.r.t. the dynamics $\cU_\infty (t;s)$. To this end, we remark that
\[ \frac{d}{dt} \left\langle \cU_\infty (t;s) \psi , \cL^2_\infty (t) \cU_\infty (t;s) \psi \right\rangle = \left\langle \cU_\infty (t;s) \psi, \left( \cL_\infty (t) \dot{\cL}_\infty (t) + \dot{\cL}_\infty (t) \cL_\infty (t) \right) \cU_\infty (t;s) \psi \right\rangle \]
Hence
\[ \begin{split}  &\left| \frac{d}{dt} \left\langle \cU_\infty (t;s) \psi , \cL^2_\infty (t) \cU_\infty (t;s) \psi \right\rangle \right|^2 \\ &\hspace{3cm} \leq \left\langle \cU_\infty (t;s) \psi , \cL^2_\infty (t) \cU_\infty (t;s) \psi \right\rangle^{1/2} \, \left\langle \cU_\infty (t;s) \psi , \dot{\cL}_\infty (t)^2 \cU_\infty (t;s) \psi \right\rangle^{1/2}
\end{split} \]
where
\[ \dot{\cL}_\infty (t) = \dot{A}_{1,t} + \dot{A}_{2,t} + \dot{A}_{3,t}\,. \]
Here we introduced the notation
\[ \begin{split}
\dot{A}_{1,t} &= \int dx \left( V * (\dot{\ph}_t \overline{\ph}_t + \ph_t \dot{\overline{\ph}}_t) \right) (x) \, a_x^* a_x \\
\dot{A}_{2,t} &= \int dx dy \, V(x-y) \left( \dot{\ph}_t (x) \overline{\ph}_t(y) + \ph_t (x) \dot{\overline{\ph}}_t (y) \right) \, a_x^* a_y
\end{split}
\]
and $\dot{A}_{3,t} = \dot{B}_t + \dot{B}^*_t$ with
\[ \dot{B}_t = 2 \int dx dy V(x-y) \dot{\ph}_t (x) \ph_t (y) \, a_x^* a_y^* \]
It follows that
\begin{equation}\label{eq:cLinfdot} \dot{\cL}^2_\infty (t) \lesssim \dot{A}^2_{1,t} + \dot{A}^2_{2,t} + \dot{A}^2_{3,t} \end{equation}
We observe that
\begin{equation}\label{eq:Vdotph} \begin{split} \sup_x \, \left| V * (\dot{\ph}_t \overline{\ph}_t + \ph_t \dot{\overline{\ph}}_t) \right| & \leq 2\, \sup_x \int dy |V(x-y)| |\dot{\ph}_t (y)| |\ph_t (y)|  \\ & \leq 2\, \| \dot{\ph}_t \|_2 \, \sup_x \left( \int dy \, V^2 (x-y) |\ph_t (y)|^2 \right)^{1/2} \\ & \lesssim D \, \| \dot{\ph}_t \|_2 \| \ph_t \|_{H^1} \lesssim C \end{split} \end{equation}
for a constant $C$ depending only on $D$ and $\| \ph \|_{H^2}$ (here we used (\ref{eq:dotph})). Eq. (\ref{eq:Vdotph}) implies, similarly to (\ref{eq:A1t}), that $\dot{A}_{1,t}^2 \leq C \cN^2$ (this time, however, with a constant $C$ depending on the $H^2$-norm of the initial data $\ph$). Analogously to (\ref{eq:A2t}), we find $\dot{A}_{2,t}^2 \leq C \cN^2$ (again with a constant depending on $\| \ph \|_{H^2}$). Finally, to bound the contribution from $\dot{A}_{3,t}^2$, we observe that
\[ \begin{split} \dot{B} \, \dot{B}^* &= \int dx dy dz dw \, V(x-y) V(z-w) \dot{\ph}_t (x) \ph_t (y) \dot{\overline{\ph}}_t (z) \overline{\ph}_t (w) a_x^* a_y^* a_z a_w \\ &\leq 2 \int dx dy dz dw \, V(x-y)^2 |\dot{\ph}_t (x)|^2 |\ph_t (y)|^2 \, a_z^* a_w^* a_w a_z \\ & \lesssim \| \ph_t \|_{H^1}^2 \, \| \dot{\ph}_t \|_2^2 \, \cN^2 \leq C \cN^2 \end{split} \]
for a constant $C$ depending on $\| \ph \|_{H^2}$. Similarly, we also obtain $\dot{B}^* \dot{B} \leq C (\cN+1)^2$ (the $+1$ factor takes into account terms arising from various commutators; this computation is similar to (\ref{eq:B*B})). We conclude that
\[ \dot{\cL}_{\infty}^2 (t) \leq C (\cN +1)^2 \]
and thus, with (\ref{eq:Ninfty}),
\[ \left| \frac{d}{dt} \left\langle \cU_\infty (t;s) \psi , \cL_\infty^2 (t) \cU_\infty (t;s) \psi \right\rangle \right| \leq C e^{K |t-s|} \langle \psi, (\cN+1)^2 \psi \rangle^{1/2} \, \left\langle \cU_\infty (t;s) \psi , \cL_\infty^2 (t) \cU_\infty (t;s) \psi \right\rangle^{1/2} \]
Gronwall's inequality implies that
\[ \left\langle \cU_\infty (t;s) \psi , \cL_\infty^2 (t) \cU_\infty (t;s) \psi \right\rangle \leq C e^{K |t-s|} \left\langle \psi, \left( \cL_\infty^2 (0) + \cN^2 + 1 \right) \psi \right\rangle \]
for appropriate constants $C,K >0$. {F}rom (\ref{eq:K2L2}) and (\ref{eq:L2K2}), we conclude that
\[ \begin{split} \left\langle \cU_\infty (t;s) \psi , \cK^2 \cU_\infty (t;s) \psi \right\rangle & \leq C e^{K |t-s|} \left\langle \psi, \left( \cL_\infty^2 (0) + \cN^2 + 1 \right) \psi \right\rangle \\
& \leq C e^{K |t-s|} \left\langle \psi, \left( \cK^2 + \cN^2 + 1 \right) \psi \right\rangle
\end{split}  \]
\end{proof}

Proposition \ref{prop:2} can be used to prove properties of the Bogoliubov transformation $\Theta (t;s)$ defined in (\ref{eq:thetat}). In particular, we can now show Theorem \ref{thm:bog}.

\begin{proof}[Proof of Theorem \ref{thm:bog}]
It is proven in \cite{CLS}[Lemma 8.1] that
\[ \begin{split} P_k \, \cU^*_\infty (t;s) a^* (f) \cU_\infty (t;s) \Omega &= 0 \\ P_k \, \cU^*_\infty (t;s) a (f) \cU_\infty (t;s) \Omega &= 0  \end{split} \]
for all $k \not = 1$ (here $P_k$ denotes the orthogonal projection on the $k$-th sector $\cF_k$ of the Fock space). In other words, the vector $\cU^*_\infty (t;s) a^* (f) \cU_\infty (t;s) \Omega$ lives in the one-particle sector, for all $t,s \in \bR$. This implies that there are linear operators $U (t;s), V (t;s) : L^2 (\bR^3) \to L^2 (\bR^3)$ with
\[ \begin{split} \cU^*_\infty (t;s) a^* (f) \cU_\infty (t;s) \Omega =& a^* (U (t;s) f) \Omega \\ \cU^*_\infty (t;s) a (f) \cU_\infty (t;s) \Omega =& a^* (J V (t;s) f) \Omega \end{split} \]
Recall here that $Jf = \overline{f}$. The operators $U(t;s)$ and $V(t;s)$ are bounded because, by Proposition \ref{prop:2},
\[ \begin{split}
\| U (t;s) f \| & = \| a^* (U(t;s) f) \Omega \| = \| a^* (f) \cU_\infty (t;s) \Omega \| \leq \| f \| \, \| (\cN+1)^{1/2} \cU_\infty (t;s) \Omega \| \leq C e^{K|t-s|} \| f \| \end{split} \]
and, similarly,
\[ \begin{split} \| V (t;s) f \| &= \| a^* (J V (t;s) f) \Omega \| = \| a (f) \cU_\infty (t;s) \Omega \| \leq \| f \| \, \| \cN^{1/2} \cU_\infty (t;s) \Omega \| \leq C e^{K |t-s|} \| f \| \, . \end{split} \]
Defining the linear map $\Theta (t;s) : L^2 (\bR^3) \oplus L^2 (\bR^3) \to L^2 (\bR^3) \oplus L^2 (\bR^3)$ as in (\ref{eq:thetat}), we obtain
\begin{equation}\label{eq:AOmega} \cU^*_\infty (t;s) A(f,g) \cU_\infty (t;s) \Omega = A(\Theta (t;s) (f,g))  \Omega \end{equation}
for all $f,g \in L^2 (\bR^3)$. Note that the boundedness of $\Theta (t;s)$ follows from the boundedness of $U (t;s)$ and $V(t;s)$. We notice also that the boundedness of $U (t;s)$ and $V(t;s)$ from $H^1 (\bR^3)$ into $H^1 (\bR^3)$ follows similarly as above, but using the estimate (\ref{eq:K2infty}) for the growth of the kinetic energy.

Next, we define, for fixed $\psi \in \cD (\cK+\cN)$ (this assumption guarantees that $\cU_\infty (t;s)$ can be differentiated; see Proposition \ref{prop:GV}), $s \in \bR$, $g \in L^2 (\bR^3)$, and for any bounded operator $M$ on $\cF$, with $M \cD (\cK+\cN) \subset \cD (\cK+\cN)$,
\[ F(t) := \sum_{\sharp} \sup_{f \in L^2 (\bR^3)} \frac{\left\| \left[ \left[ \cU^*_\infty (t;s) a^{\sharp} (f) \cU_\infty (t;s) , a^{\flat} (g) \right] , M \right] \psi \right\|}{\| f \|} \]
where $a^\sharp , a^\flat$ are either annihilation or creation operators. We observe that, since $e^{-i \cK t} a^\sharp (f) e^{i\cK t} = a^\sharp (e^{it\Delta} f)$, and since $\|e^{it\Delta} f \| = \| f \|$, we also have
\[ F(t) := \sum_{\sharp} \sup_{f \in L^2 (\bR^3)} \frac{\left\| \left[ \left[ \cU^*_\infty (t;s) e^{-i \cK (t-s)} a^{\sharp} (f) \, e^{i\cK (t-s)} \, \cU_\infty (t;s) , a^{\flat} (g) \right] , M \right] \psi \right\|}{\| f \|} \]
A simple computation shows that $F(s) = 0$. Moreover, we find
\[ \begin{split} \frac{d}{dt} \,& \left[ \left[ \cU^*_\infty (t;s) e^{-i \cK (t-s)} a (f) \, e^{i\cK (t-s)} \, \cU_\infty (t;s) , a^{\flat} (g) \right] , M \right] \psi \\ &=  \left[ \left[ \cU^*_\infty (t;s) \left[ (\cL_\infty (t) - \cK), e^{-i \cK (t-s)} a (f) \, e^{i\cK (t-s)} \right] \, \cU_\infty (t;s) , a^{\flat} (g) \right] , M \right] \psi \\ &=  \left[ \left[ \cU^*_\infty (t;s) \left[ (\cL_\infty (t) - \cK), a (e^{-i\Delta (t-s)} f) \right] \, \cU_\infty (t;s) , a^{\flat} (g) \right] , M \right] \psi
\end{split} \]
Using the canonical commutation relations, and the notation $f (t) = e^{-it\Delta} f$, it is simple to check that
\[ [ \cL_\infty (t) - \cK , a(f (t-s)) ] = a ( (V*|\ph_t|^2) f(t-s) + (V*f(t-s) \overline{\ph}_t)\ph_t) + a^* (2 (V*\overline{f} (t-s) \ph_t )\ph_t ) \, . \]
Hence, we conclude that
\[ \begin{split}  & \left\| \left[ \left[ \cU^*_\infty (t;s) e^{-i \cK (t-s)} a (f) \, e^{i\cK (t-s)} \, \cU_\infty (t;s) , a^{\flat} (g) \right] , M \right] \psi \right\| \\ & \hspace{.5cm} \leq \int_s^t d\tau \, \left\| \left[ \left[ \cU^*_\infty (\tau;s) a \left( (V*|\ph_\tau|^2)f (\tau-s) + (V*f(\tau-s) \overline{\ph}_\tau)\ph_\tau \right) \, \cU_\infty (\tau;s) , a^{\flat} (g) \right] , M \right] \psi \right\| \\ & \hspace{1cm} + \int_s^t d\tau \, \left\| \left[ \left[ \cU^*_\infty (\tau;s) a^* \left(2(V*\overline{f} (\tau-s) \ph_\tau)\ph_\tau \right) \, \cU_\infty (\tau;s) , a^{\flat} (g) \right] , M \right] \psi \right\| \\ & \hspace{.5cm} \leq C \int_s^t d\tau \, \left(\| (V*|\ph_\tau|^2)f (\tau-s) + (V*f(\tau-s) \overline{\ph}_\tau)\ph_\tau \| + \| (V*\overline{f}(\tau -s )\ph_\tau)\ph_\tau \| \right) \, F (\tau)
\end{split} \]
Notice that, under the assumption $V^2 \leq D (1-\Delta)$, we find
\[ \| (V*|\ph_\tau|^2) f (\tau-s) \| \leq \| f (\tau-s) \| \, \sup_x \int dy V(x-y) |\ph_\tau (y)|^2 \leq C \| f \| \| \ph_\tau \|_{H^{1}}^2 \leq C \| f \| \]
and
\[  \| (V* f (\tau-s) \overline{\ph}_\tau) \ph_\tau \| \leq \| \ph_\tau \| \,  \sup_x \int dy \, V(x-y) |f(\tau-s, y)| | \ph_\tau (y)|
\leq C \| f \| \, \| \ph_\tau \|_{H^1}^2 \leq C \| f \| \]
for a constant $C$, independent of $t$ and $f$. Therefore
\[ \frac{\left\| \left[ \left[ \cU^*_\infty (t;s) e^{-i \cK (t-s)} a (f) \, e^{i\cK (t-s)} \, \cU_\infty (t;s) , a^{\flat} (g) \right] , M \right] \psi \right\|}{\| f \|} \leq C \int_s^t d\tau \, F (\tau) \]
The term with $a^* (f)$ instead of $a(f)$ can be bounded similarly. Therefore, taking the supremum over $f \in L^2 (\bR^3)$, we conclude that
\[ F (t) \leq \int_s^t \rd \tau F(\tau) \]
Since $F(t) \geq 0$ for all $t \in \bR$, and $F(s) = 0$, we conclude that $F(t) = 0$ for all $t \in \bR$. This proves that
\begin{equation}\label{eq:2comm} \left[ \left[ \cU_\infty^* (t;s) A(f_1,g_1) \cU_\infty (t;s) , A(f_2,g_2) \right] , M \right] = 0 \end{equation}
for every $f_1,f_2,g_1,g_2 \in L^2 (\bR^3)$, and every bounded operator $M$ on $\cF$ such that $M \cD (\cK+\cN) \subset \cD (\cK+\cN)$.
This implies that
\begin{equation}\label{eq:1comm} \langle \psi , \left[ \cU_\infty^* (t;s) A(f_1,g_1) \cU_\infty (t;s) , A(f_2, g_2) \right] \psi \rangle  = \langle \Omega , \left[ \cU_\infty^* (t;s) A(f_1,g_1) \cU_\infty (t;s) , A(f_2, g_2) \right] \Omega \rangle \end{equation}
for all $\psi \in \cD (\cK+\cN)$ with $\| \psi \| =1$. This follows because, from (\ref{eq:2comm}),
\[ \left[ \left[ \cU_\infty^* (t;s) A(f_1,g_1) \cU_\infty (t;s) , A(f_2, g_2) \right] , P_\psi \right] = \left[ \left[ \cU_\infty^* (t;s) A(f_1,g_1) \cU_\infty (t;s) , A (f_2, g_2) \right] , P_\Omega \right] = 0 \]
where $P_\psi, P_\Omega$ denote the orthogonal projections into $\psi$ and, respectively, $\Omega$. Therefore,
\[ \begin{split}\langle \psi, \left[ \cU_\infty^* (t;s) A(f_1,g_1) \cU_\infty (t;s) , A(f_2, g_2) \right] \Omega \rangle &= \langle \psi, \left[ \cU_\infty^* (t;s) A(f_1,g_1) \cU_\infty (t;s) , A(f_2, g_2)  \right] P_\psi \Omega \rangle \\ &= \langle \psi, \left[ \cU_\infty^* (t;s) A(f_1,g_1) \cU_\infty (t;s) , A(f_2, g_2) \right] \psi \rangle \, \langle \psi, \Omega \rangle \end{split} \]
and, similarly,
\[ \langle \psi, \left[ \cU_\infty^* (t;s) A(f_1,g_1) \cU_\infty (t;s) , A(f_2,g_2) \right] \Omega \rangle =  \langle \psi, \Omega \rangle \,\langle \Omega, \left[ \cU_\infty^* (t;s) A(f_1,g_1) \cU_\infty (t;s) , A(f_2,g_2) \right] \Omega \rangle \]
If $\langle \psi, \Omega \rangle \not = 0$, this implies immediately (\ref{eq:1comm}). If $\langle \psi, \Omega \rangle = 0$, (\ref{eq:1comm}) follows repeating the argument with $\wt{\psi} = 2^{-1/2} (\psi + \Omega)$.
Combining (\ref{eq:1comm}) with (\ref{eq:AOmega}), we find
\[\begin{split} \langle \psi, \left[ \cU_\infty^* (t;s) A(f_1,g_1) \cU_\infty (t;s) , A(f_2,g_2)  \right] \psi \rangle &= \langle \Omega, \left[ A (\Theta (t;s) (f_1, g_1)) , A (f_2, g_2) \right] \Omega \rangle\\ & = \left( \Theta (t;s) (f_1, g_1) , S (f_2, g_2) \right)_{L^2 \oplus L^2} \end{split} \]
where $S$ is defined as in (\ref{eq:bog22}). Hence
\begin{equation}\label{eq:comm} \left[ \cU_\infty^* (t;s) A(f_1,g_1) \cU_\infty (t;s) - A(\Theta (t;s) (f_1, g_1)) , A(f_2,g_2)  \right] = 0 \end{equation}
for all $f_1, g_1,f_2,g_2 \in L^2 (\bR^3)$. Let $R = \cU_\infty^* (t;s) A(f_1,g_1) \cU_\infty (t;s) - A(\Theta (t;s) (f_1, g_1))$. We already proved in (\ref{eq:AOmega}) that $R \Omega = 0$, and in (\ref{eq:comm}) that $R$ commutes with any creation and annihilation operator. Since states of the form $a^* (h_1) \dots a^* (f_n) \Omega$ form a basis for $\cF$, this implies that $R =0$, and therefore that
\[  \cU_\infty^* (t;s) A(f,g) \cU_\infty (t;s) = A(\Theta (t;s) (f, g)) \]
for all $f,g \in L^2 (\bR^3)$. {F}rom
\[ \begin{split}
\left( A (\Theta (t;s) (f,g)) \right)^* &= \left(   \cU_\infty^* (t;s) A (f,g) \cU_\infty (t;s) \right)^* \\ &=  \cU_\infty^* (t;s) A^* (f,g) \cU_\infty (t;s) \\ &=
\cU_\infty^* (t;s) A (Jg,Jf) \cU_\infty (t;s) = A(\Theta (t;s) (Jg,Jf))
\end{split} \]
we deduce (\ref{eq:bog1}). The property (\ref{eq:bog2}), on the other hand, follows because
\[ \begin{split}
\left[ A (\Theta (t;s) (f_1, g_1)) , A (\Theta (t;s) (f_2, g_2)) \right]  & = \left[  \cU_\infty^* (t;s) A (f_1,g_1) \cU_\infty (t;s) , \cU_\infty^* (t;s) A (f_2,g_2) \cU_\infty (t;s) \right] \\ & = \cU_\infty^* (t;s)
\left[ A (f_1, g_1) , A (f_2, g_2) \right] \cU_\infty (t;s) \\ & = \left( (f_1, g_1), S (f_2 , g_2) \right)_{L^2 \oplus L^2} \end{split} \]

Finally, to prove (\ref{eq:Ttph}), we observe that
\begin{equation}\label{eq:Thetaph} A (\Theta (t;s) (\ph_t , \overline{\ph}_t)) = \cU^*_\infty (t;s) A (\ph_t , \overline{\ph}_t) \cU_\infty (t;s) \end{equation}
We have
\[ \begin{split} i \frac{d}{dt} \cU^*_\infty (t;s) A (\ph_t , \overline{\ph}_t) \cU_\infty (t;s) = \; &- \cU^*_\infty (t;s) \left[ \cL_\infty (t) , A (\ph_t , \overline{\ph}_t) \right] \cU_\infty (t;s) \\ &+ \cU^*_\infty (t;s) A (i \dot{\ph}_t, \overline{-i\dot{\ph}_t} ) \cU_\infty (t;s) \end{split} \]
A simple computation shows that
\[ \left[ \cL_\infty (t) , A (\ph_t, \overline{\ph}_t) \right] = A(-\Delta \ph_t + (V*|\ph_t|^2) \ph_t , \Delta \overline{\ph}_t - (V*|\ph_t|^2)\overline{\ph}_t ) \]
and therefore that
\[ \frac{d}{dt}\, \cU^*_\infty (t;s) A (\ph_t , \overline{\ph}_t) \cU_\infty (t;s) = 0 \, .\]
Hence
\[ \cU^*_\infty (t;s) A (\ph_t , \overline{\ph}_t) \cU_\infty (t;s) = A (\ph_s , \overline{\ph}_s) \]
and thus, from (\ref{eq:Thetaph}),
\[ \Theta (t;s) (\ph_t, \overline{\ph}_t) = (\ph_s , \overline{\ph}_s) \] \end{proof}

Next, we control the growth of the number of particles operator and of the kinetic energy operator with respect to the full fluctuation evolution $\cU (t;s)$, defined in (\ref{eq:cU}).
\begin{proposition}\label{prop:1}
For every $j \in \bN$, there exist constants $C, K >0$ (depending only on $j$, on $D$ and on $\| \ph \|_{H^1}$) such that
\begin{equation}\label{eq:prop1-N} \left\langle \cU (t;s) \psi , \cN^j \, \cU (t;s) \psi \right\rangle \leq C e^{K |t-s|} \left\langle \psi, (\cN +1)^{2j+2} \, \psi \right\rangle \end{equation}
for all $\psi \in \cF$, and all $t,s \in \bR$. Moreover, there are constants $C,K$ (depending only on $D$ and on $\| \ph \|_{H^2}$) such that
\begin{equation}\label{eq:prop1-K} \left\langle \cU (t;s) \psi , \cK \, \cU (t;s) \psi \right\rangle \leq C e^{K |t-s|} \left\langle \psi, (\cK + \cN^8 +1) \, \psi \right\rangle \end{equation}
\end{proposition}
\begin{proof}
Eq. (\ref{eq:prop1-N}) is taken from \cite{CLS}[Proposition 5.1]. We show next the bound (\ref{eq:prop1-K}) on the growth of the kinetic energy. To this end, we remark that, by Lemma \ref{lm:KHN}, there exists a constant $C$ (depending on $D$, $\| \ph \|_{H^1}$) such that
\begin{equation} \label{eq:prop11} \begin{split} \left\langle \cU (t;s) \psi , \cK \, \cU (t;s) \psi \right\rangle \leq \; & 2 \left\langle \cU (t;s) \psi , \cL(t) \, \cU (t;s) \psi\right\rangle + C \left\langle \cU (t;s) \psi , (\cN+1) \left( 1+ \frac{\cN^2}{N^2} \right) \cU (t;s) \psi \right\rangle \\ \leq \; & 2 \left\langle \cU (t;s) \psi , \cL(t) \, \cU (t;s) \psi\right\rangle + C e^{K|t-s|} \left\langle \psi , (\cN+1)^8 \psi \right\rangle \end{split}\end{equation}
In the second inequality, we used (\ref{eq:prop1-N}). To bound the first term, we observe that
\[ \frac{d}{dt} \left\langle \cU (t;s) \psi , \cL(t) \, \cU (t;s) \psi\right\rangle = \left\langle \cU (t;s) \psi , \dot{\cL}(t) \, \cU (t;s) \psi\right\rangle \]
where
\begin{equation}\label{eq:cLdot} \begin{split}
\dot{\cL} (t) = \; & \int dx \, \left(V* (\dot{\ph}_t \overline{\ph}_t + \ph_t \dot{\overline{\ph}}_t ) \right) (x) \, a^*_x a_x \\ &+\int dx dy \, V(x-y) \left( \dot{\ph}_t (x) \overline{\ph}_t (y) + \ph_t (x) \dot{\overline{\ph}}_t (y) \right) \, a_x^* a_y \\ & +2 \int dx dy \, V(x-y) \left( \dot{\ph}_t (x) \ph_t (y) \, a_x^* a_y^* + \dot{\overline{\ph}}_t (x) \overline{\ph}_t (y) a_x a_y \right) \\
&+ \frac{1}{\sqrt{N}} \int dx dy \, V(x-y) a_x^* \left( \dot{\ph}_t (y) a^*_y + \dot{\overline{\ph}}_t (y) a_y \right) a_x
\end{split} \end{equation}
To bound the last term, we observe that, for any $\psi \in \cF$,
\[ \begin{split}
\Big| \frac{1}{\sqrt{N}} \int dx dy \, V(x-y) &\dot{\ph}_t (y) \, \left\langle \psi, a_x^*  a^*_y  a_x \psi \right\rangle \Big| \\ \leq \; &
\frac{1}{\sqrt{N}} \int dx dy \, |V(x-y)| \, |\dot{\ph}_t (y)| \, \| a_x a_y \psi \| \, \| a_x \psi \| \\
\leq \; & \int dx dy \, |V (x-y)| |\dot{\ph}_t (y)|^2 \| a_x \psi \|^2 + \frac{1}{N} \int dx dy \, |V(x-y)| \, \| a_y a_x \psi \|^2 \\ \leq\; &
\| \dot{\ph}_t \|^2  \, \| \cK^{1/2} \psi \|^2 + \left\langle \psi,
\left( \cK + C \cN \left(1+ \frac{\cN^2}{N^2} \right) \right) \psi \right\rangle \\
\leq \; &C \left\langle \psi,
\left( \cK + C \cN \left(1+ \frac{\cN^2}{N^2} \right) \right) \psi \right\rangle
\end{split} \]
for a constant $C$ depending on $\| \ph \|_{H^2}$. Here we used the bounds (\ref{eq:cW}) and, in the last step, (\ref{eq:dotph}). The other terms on the r.h.s. of (\ref{eq:cLdot}) are just the time derivative $\dot{\cL}_\infty (t)$ of the generator $\cL_\infty (t)$ of the limiting dynamics. {F}rom the proof of Proposition \ref{prop:2}, we have $\dot{\cL}_\infty (t) \leq C (\cN+1)$. Hence, we conclude that
\[ \begin{split}
\left| \frac{d}{dt} \left\langle \cU (t;s) \psi , \cL(t) \, \cU (t;s) \psi\right\rangle  \right| \leq \; & C \left\langle \cU (t;s) \psi ,
\left( \cK + (\cN + 1) \left(1+ \frac{\cN^2}{N^2} \right) \right) \, \cU (t;s) \psi\right\rangle \\ \leq \; &  C \left\langle \cU (t;s) \psi ,
\left( \cL (t) + (\cN + 1) \left(1+ \frac{\cN^2}{N^2} \right) \right) \, \cU (t;s) \psi\right\rangle \\ \leq \; &  C \left\langle \cU (t;s) \psi , \cL (t) \, \cU (t;s) \psi\right\rangle + C e^{K |t-s|} \langle \psi , (\cN+1)^8 \psi \rangle \end{split} \]
From Gronwall lemma, we find that
\[ \begin{split}
\left| \left\langle \cU (t;s) \psi , \cL(t) \, \cU (t;s) \psi\right\rangle \right| \leq C e^{K |t-s|} \langle \psi, \left( \cK + \cN^8 + 1 \right) \psi \rangle \end{split} \]
Inserting this estimate in (\ref{eq:prop11}), we obtain (\ref{eq:prop1-K}).
\end{proof}

Compared with Proposition \ref{prop:2}, in Proposition \ref{prop:1} we do not have control on the growth of the square of the kinetic energy operator. Due to the cubic and quartic terms in the generator of $\cU (t;s)$, it is not clear whether a bound of the form (\ref{eq:K2infty}) holds with $\cU_\infty (t;s)$ replaced by $\cU (t;s)$. Fortunately, we only need some control on the growth of products of the form $(\cN+1)^j \cK (\cN + 1)^j$. For these operators, weak $N$-dependent bounds are established in the next lemma going back to the original definition of the fluctuation evolution $\cU (t;s) = W^* (\sqrt{N} \ph_t) e^{-i\cH_N (t-s)} W (\sqrt{N} \ph_s)$ (stronger, $N$-independent, bounds for the growth of operators of the form $(\cN+1)^j \cK (\cN + 1)^j$ with respect to the limiting dynamics $\cU_\infty (t;s)$ follow from Proposition \ref{prop:2} since, by Cauchy-Schwarz, $(\cN+1)^j \cK (\cN + 1)^j \leq \cK^2 + (\cN+1)^{4j}$).
\begin{lemma}\label{lm:NKN}
For every $j\in \bN$, there exists a constant $C$ such that
\[ \begin{split}  \langle &\psi , \cU^* (t;s) (\cN +1)^j \, ( \cK +1) \, (\cN +1)^j \cU (t;s) \psi \rangle \\ &\leq C \left[ (N+1)^{2j} \, \langle \psi, \cU^* (t;s) \, \cK \, \cU (t;s) \psi \rangle + \langle \psi, (\cN + N +1)^j \, \left( \cK + \cN^2 + 1 \right) \, (\cN + N +1)^j \, \psi \rangle \right] \end{split} \]
\end{lemma}
\begin{proof}
We proceed by induction over $j \in \bN$. For $j=0$, the claim is clear. Now we assume it to hold for $j= n-1$, and we prove if for $j=n$. To this end, we observe that, from (\ref{eq:shift}),
\begin{equation}\label{eq:WNW}
W^* (\sqrt{N} \ph_t ) \cN W (\sqrt{N} \ph_t) = (\cN - \sqrt{N} A(\ph_t , \overline{\ph}_t) + N)
\end{equation}
where we recall the notation $A(f,g) = a(f) + a^*(\overline{g})$. Therefore, fixing for simplicity $s=0$, we find
\[ \begin{split}
\Big\langle \psi, &\cU^* (t;0) \, (\cN+1)^j \, (\cK+1) \, (\cN +1)^j \, \cU (t;0) \psi \Big\rangle \\ = \; & \left\langle \psi, W^* (\sqrt{N} \ph) e^{i\cH_N t} \, W (\sqrt{N} \ph_t) \, (\cN+1)^j \, (\cK+1) \, (\cN +1)^j \, W^* (\sqrt{N} \ph_t) \, e^{i \cH_N t} \, W (\sqrt{N} \ph) \psi \right\rangle \\ =\; & \Big\langle \psi, W^* (\sqrt{N} \ph) e^{i\cH_N t} \, \left(\cN - \sqrt{N} A(\ph_t,\overline{\ph}_t) + N \right) W (\sqrt{N} \ph_t) \, (\cN+1)^{j-1} \, (\cK+1) \, (\cN +1)^{j-1} \, \\ & \hspace{5cm} \times W^* (\sqrt{N} \ph_t) \left(\cN -\sqrt{N} A(\ph_t , \overline{\ph}_t) + N \right) \, e^{i \cH_N t} \, W (\sqrt{N} \ph) \psi \Big\rangle
\end{split} \]
By Cauchy-Schwarz, we have
\begin{equation}\label{eq:KNK-1} \begin{split}
\Big\langle \psi, &\cU^* (t;0) \,  (\cN+1)^j \, (\cK+1) \, (\cN +1)^j \, \cU (t;0) \psi \Big\rangle \\ \lesssim \; & \left\langle \psi, W^* (\sqrt{N} \ph) e^{i\cH_N t} \, \cN  W (\sqrt{N} \ph_t) \, (\cN+1)^{j-1} \, (\cK+1) \, (\cN +1)^{j-1} \right. \\ &\hspace{6cm} \left. \times  W^* (\sqrt{N} \ph_t) \cN \, e^{i \cH_N t} \, W (\sqrt{N} \ph) \psi \right\rangle \\ &+ N \Big\langle \psi, W^* (\sqrt{N} \ph) e^{i\cH_N t} \, A(\ph_t , \overline{\ph}_t) \,  W (\sqrt{N} \ph_t) \, (\cN+1)^{j-1} \, (\cK+1) \, (\cN +1)^{j-1} \\ &\hspace{6cm} \times  W^* (\sqrt{N} \ph_t) A (\ph_t , \overline{\ph}_t)  \, e^{i \cH_N t} \, W (\sqrt{N} \ph) \psi \Big\rangle \\ & +N^2 \left\langle \psi, \cU^* (t;0) \, (\cN+1)^{j-1} \, (\cK+1) \, (\cN +1)^{j-1} \, \cU (t;0) \psi \right\rangle \\ = \; & \text{I} + \text{II} + \text{III}
\end{split} \end{equation}
In the last term, we can directly apply the induction assumption. We find
\begin{equation}\label{eq:III}
\begin{split}
\text{III} \leq C N^2 \Big[ &(N+1)^{2(j-1)} \, \langle \psi, \cU (t;s) \, \cK \, \cU (t;s) \psi \rangle \\ &+ \langle \psi, (\cN + N +1)^{j-1} \, \left(\cK + \cN^2 + 1 \right) \, (\cN + N +1)^{j-1} \, \psi \rangle \Big] \\ \leq C \Big[ (N&+1)^{2j}  \, \langle \psi, \cU (t;s) \, \cK \, \cU (t;s) \psi \rangle + \langle \psi, (\cN + N +1)^{j} \, \left( \cK + \cN^2 + 1 \right) \, (\cN + N +1)^{j} \, \psi \rangle \Big] \end{split}\end{equation}
To bound the first term on the r.h.s. of (\ref{eq:KNK-1}), on the other hand, we use the fact that $\cN$ commutes with $\cH_N$. {F}rom (\ref{eq:WNW}), we find
\begin{equation}\label{eq:NKNI} \begin{split}
\text{I} = \; &\left\langle (\cN + \sqrt{N} A (\ph,\overline{\ph}) + N) \psi,  \, \cU^* (t;0) \, (\cN+1)^{j-1} \, (\cK+1) \right. \\ & \left. \hspace{5cm} \times  (\cN +1)^{j-1} \, \cU (t;0) (\cN + \sqrt{N} A (\ph,\overline{\ph}) + N) \, \psi \right\rangle  \\ \lesssim \; &\left\langle \cN \psi ,  \, \cU^* (t;0) \, (\cN+1)^{j-1} \, (\cK+1) \, (\cN +1)^{j-1} \, \cU (t;0) \, \cN \, \psi \right\rangle \\ &+ N \left\langle A (\ph,\overline{\ph}) \psi,  \, \cU^* (t;0) \, (\cN+1)^{j-1} \, (\cK+1) \, (\cN +1)^{j-1} \, \cU (t;0) \, A(\ph,\overline{\ph}) \psi \right\rangle \\
&+ N^2 \left\langle  \psi,  \, \cU^* (t;0) \, (\cN+1)^{j-1} \, (\cK+1) \, (\cN +1)^{j-1} \, \cU (t;0) \, \psi \right\rangle
\\ = \; & \text{A}_1 + \text{A}_2 + \text{A}_3 \end{split} \end{equation}
The term $\text{A}_3$ can be handled like the term $\text{III}$ in (\ref{eq:III}). To bound $\text{A}_1$, we use the induction assumption. We find
\begin{equation}\label{eq:IA1} \begin{split}
\text{A}_1 \leq \; & C \Big[ (N+1)^{2(j-1)} \, \left\langle \cN \psi , \cU^* (t;0) \, \cK \, \cU (t;0) \, \cN \, \psi \right\rangle \\ &\hspace{1cm}+ \left\langle \cN \psi, (\cN+N+1)^{j-1} (\cK + \cN^2 + 1) (\cN+N+1)^{j-1} \cN \psi \right\rangle \Big] \\ \leq \; & C (N+1)^{2(j-1)} \, \left\langle \cN \psi , W^* (\sqrt{N} \ph) e^{i\cH_N t} \, \cK  \, e^{-i\cH_N t} W (\sqrt{N} \ph) \, \cN \, \psi \right\rangle \\ &+ C \left\langle \psi, (\cN+N+1)^{j} (\cK + \cN^2 + 1) (\cN+N+1)^{j} \cN \psi \right\rangle
\end{split} \end{equation}
where we used that, by (\ref{eq:shift}),
\[ W^* (\sqrt{N} \ph_t ) \cK \, W (\sqrt{N} \ph_t) = (\cK + \sqrt{N} A(\Delta \ph_t , \Delta \overline{\ph}_t) + N \| \nabla \ph_t \|^2) \leq \cK + N \| \nabla \ph_t \|^2 \leq \cK + C N \,.  \]
To bound the first term on the r.h.s. of (\ref{eq:IA1}), we use Lemma \ref{lm:KHN}. We find
\[ \begin{split}
\Big\langle \cN \psi , W^* (\sqrt{N} \ph) e^{i\cH_N t} \,& \cK  \, e^{-i\cH_N t} W (\sqrt{N} \ph) \, \cN \, \psi \Big\rangle \\ \lesssim \; &
\left\langle \cN \psi , W^* (\sqrt{N} \ph)  \, \left( \cH_N + C N^{-2} \cN^2 +1 \right)  \, W (\sqrt{N} \ph) \, \cN \, \psi \right\rangle \\ \lesssim \; & \, \left\langle \cN \psi , W^* (\sqrt{N} \ph)  \, \left( \cK + C N^{-2} \cN^2 +1 \right)  \, W (\sqrt{N} \ph) \, \cN \, \psi \right\rangle  \\ \leq \; & C \, \left\langle \cN \psi , \, \left( \cK + N+1 + N^{-2} \cN^2 \right)  \, \cN \, \psi \right\rangle
\end{split} \]
which implies that
\[ \begin{split}
\text{A}_1 \leq \; & C \, \left\langle \psi , \, (\cN + N+1)^j \left( \cK + \cN^2 + 1 \right)  \, (\cN + N +1)^j \, \psi \right\rangle
\end{split} \]
To estimate the term $\text{A}_2$ on the r.h.s. of (\ref{eq:NKNI}), we proceed similarly. {F}rom the induction assumption, we find
\begin{equation}\label{eq:IA2} \begin{split}
\text{A}_2 \leq \; & C \Big[ (N+1)^{2j-1} \, \left\langle A(\ph, \overline{\ph}) \psi , \cU^* (t;0) \, \cK \, \cU (t;0) \,  A(\ph, \overline{\ph})  \, \psi \right\rangle \\ &\hspace{1cm}+ N \left\langle A(\ph, \overline{\ph}) \, \psi, (\cN+N+1)^{j-1} (\cK + \cN^2 + 1) (\cN+N+1)^{j-1} \,  A(\ph, \overline{\ph})  \psi \right\rangle \Big] \\ \leq \; & C (N+1)^{2j-1} \, \left\langle  A(\ph, \overline{\ph})  \psi , W^* (\sqrt{N} \ph) e^{i\cH_N t} \, (\cK +CN) \, e^{-i\cH_N t} W (\sqrt{N} \ph) \,  A(\ph, \overline{\ph})  \, \psi \right\rangle \\ &+ C N \left\langle  A(\ph, \overline{\ph}) \psi, (\cN+N+1)^{j-1} (\cK + \cN^2 + 1) (\cN+N+1)^{j-1}  A(\ph, \overline{\ph}) \psi \right\rangle
\end{split} \end{equation}
Using Lemmas \ref{lm:phiNK} and \ref{lm:KHN}, we conclude that
\[ \text{A}_2 \leq \;  C \left\langle \psi, (\cN+N+1)^{j} (\cK + \cN^2 + 1) (\cN+N+1)^{j} \psi \right\rangle \]
Hence
\[\text{I} \leq C \Big[ (N+1)^{2j}  \, \langle \psi, \cU (t;s) \, \cK \, \cU (t;s) \psi \rangle + \langle \psi, (\cN + N +1)^{j} \, \left( \cK + \cN^2 + 1 \right) \, (\cN + N +1)^{j} \, \psi \rangle \Big]\]
Finally, we estimate the term $\text{II}$ on the r.h.s. of (\ref{eq:KNK-1}). We find
\[ \begin{split}
\text{II} \leq \; & N \Big\langle \psi, W^* (\sqrt{N} \ph) \, e^{i\cH_N t} \, a^* (\ph_t) \,  W (\sqrt{N} \ph_t) \, (\cN+1)^{j-1} \, (\cK+1) \, (\cN +1)^{j-1}  \\ & \hspace{6cm} \times W^* (\sqrt{N} \ph_t) a (\ph_t) \, e^{i \cH_N t} \, W (\sqrt{N} \ph) \psi \Big\rangle \\ &+ N \Big\langle \psi, W^* (\sqrt{N} \ph) e^{i\cH_N t} \, a (\ph_t) \,  W (\sqrt{N} \ph_t) \, (\cN+1)^{j-1} \, (\cK+1) \, (\cN +1)^{j-1}  \\ & \hspace{6cm} W^* (\sqrt{N} \ph_t) a^* (\ph_t) \, e^{i \cH_N t} \, W (\sqrt{N} \ph) \psi \Big\rangle
\\ \leq \; & N \Big\langle \psi, \cU^* (t;0) (a^* (\ph_t) + \sqrt{N}) \, (\cN+1)^{j-1} \, (\cK+1) \, (\cN +1)^{j-1}  (a (\ph_t) + \sqrt{N}) \, \cU (t;0) \psi \Big\rangle \\ &+ N \Big\langle \psi, \cU^* (t;0) \, (a (\ph_t) -\sqrt{N}) (\cN+1)^{j-1} \, (\cK+1) \, (\cN +1)^{j-1}  (a^* (\ph_t) -\sqrt{N}) \, \cU (t;0) \psi \Big\rangle
\\ \leq \; & N^2 \Big\langle \psi, \cU^* (t;0) (\cN+1)^{j-1} \,
(\cK+1) \, (\cN +1)^{j-1} \, \cU (t;0) \psi \Big\rangle \\ &+ N
\Big\langle \psi, \cU^* (t;0) (\cN+1)^{j-1} \,
(\cK+1) \, (\cN +1)^{j} \, \cU (t;0) \psi \Big\rangle
\end{split} \]
where, in the last inequality, we used Lemma \ref{lm:phiNK}. By Cauchy-Schwarz, we find
\[ \begin{split} \text{II} \leq \; & C N^2 \Big\langle \psi, \cU^* (t;0) (\cN+1)^{j-1} \, (\cK+1) \, (\cN +1)^{j-1} \, \cU (t;0) \psi \Big\rangle \\ &+ \frac{1}{2} \Big\langle \psi, \cU^* (t;0) (\cN+1)^{j} \,
(\cK+1) \, (\cN +1)^{j} \, \cU (t;0) \psi \Big\rangle
\end{split} \]
The first term can be bounded as in (\ref{eq:III}). {F}rom (\ref{eq:KNK-1}) we conclude that
\[ \begin{split}
\Big\langle \psi, &\cU^* (t;0) \,  (\cN+1)^j \, (\cK+1) \, (\cN +1)^j \, \cU (t;0) \psi \Big\rangle \\ \lesssim \; &
C \Big[ (N+1)^{2j}  \, \langle \psi, \cU (t;s) \, \cK \, \cU (t;s) \psi \rangle + \langle \psi, (\cN + N +1)^{j} \, \left( \cK + \cN^2 + 1 \right) \, (\cN + N +1)^{j} \, \psi \rangle \Big] \\ &+\frac{1}{2} \Big\langle \psi, \cU^* (t;0) \,  (\cN+1)^j \, (\cK+1) \, (\cN +1)^j \, \cU (t;0) \psi \Big\rangle
\end{split} \]
This concludes the proof of the lemma.
\end{proof}

We conclude this section with two technical lemmas, which played an important role in the proof of Propositions \ref{prop:2}, \ref{prop:1} and of Lemma \ref{lm:NKN}. The first result is used to compare the kinetic energy operator $\cK$ with the generators $\cH_N$, $\cL(t)$ and $\cL_\infty (t)$ (these bounds are important, because it is much easier to bound the growth of the expectation of the generator than the growth of the kinetic energy).
\begin{lemma}\label{lm:KHN}
Suppose that $V^2 (x) \leq D (1-\Delta)$ and let $\cH_N$ be the Hamiltonian defined in (\ref{eq:fock-ham2}), $\cL(t)$ the generator (\ref{eq:L}) of the fluctuation dynamics $\cU (t;s)$ defined in (\ref{eq:cU}) and $\cL_\infty (t)$ the generator (\ref{eq:Linfty}) of the limiting dynamics $\cU_\infty (t;s)$ defined in (\ref{eq:Uinfty}). Then there exists a constant $C>0$ (depending only on $D$) such that
\begin{equation}\label{eq:KHN} \begin{split}
 \cK - C (\cN+1) &\leq \cL_\infty (t) \leq \cK + C (\cN+1) \, , \\
\frac{1}{2} \cK - C \cN \left(1+ \frac{\cN^2}{N^2} \right) &\leq \cH_N \leq 2 \cK + C \cN \left(1+ \frac{\cN^2}{N^2} \right) \, \qquad \text{and } \\
\frac{1}{2} \cK - C (\cN+1) \left( 1+ \frac{\cN^2}{N^2} \right) &\leq \cL (t) \leq 2\cK + C (\cN+1) \left( 1+ \frac{\cN^2}{N^2} \right) \end{split} \end{equation}
\end{lemma}

\begin{proof}
The first bound in proven in  \cite{CLS}[Lemma 6.1]. To show the second estimate, let
\[ \cW = \frac{1}{2N} \int dx dy \, |V(x-y)| \, a_x^* a_y^* a_y a_x \]
{F}rom the bound $V^2 (x) \leq D (1-\Delta)$, we obtain that, for every $\eps >0$, there exists $C_\eps >0$ with \[ |V(x)| \leq  \eps (1- \Delta) + \frac{1}{\eps} \]
Therefore, the restriction $\cW^{(n)}$ of $\cW$ on the $n$-particle sector is bounded by
\[ \cW^{(n)} = \frac{1}{N} \sum_{i<j}^n |V(x_i -x_j)| \leq \frac{\eps n}{N} \sum_{j=1}^n (1-\Delta_{x_j}) + \frac{n^2}{\eps N} \]
Choosing $\eps = \delta N/n$, we find
\[ \cW^{(n)} \leq  \delta (\cK^{(n)} + n) + \frac{n^3}{\delta N^2} \]
Hence, as operators on the Fock-space,
\begin{equation}\label{eq:cW} \cW \leq \delta (\cK + \cN) + \frac{\cN^3}{\delta N^2} \end{equation}
With $\delta = 1/2$, we conclude that
\[ \begin{split}
\cH_N &= \cK + \frac{1}{2N} \int dx dy \, V(x-y) \, a_x^* a_y^* a_y a_x \\
& \geq  \cK - \cW \geq \frac{1}{2} \cK - \frac{1}{2} \cN - 2 \frac{\cN^3}{N^2} \end{split} \]
The upper bound for $\cH_N$ can be proven analogously.
To prove the last estimate in (\ref{eq:KHN}), we observe that, by Cauchy-Schwarz,
\[ \begin{split}
\cL (t) = \; & \cL_\infty (t) +\frac{1}{2 \sqrt{N}} \int dx dy V (x-y) \, a_x^* \left( a_y^* \ph_t (y) + a_y \overline{\ph}_t (y) \right) a_x +\frac{1}{2N} \int dx dy \, V( x-y) \, a_x^* a_y^* a_y a_x \\ \geq \; &
 \cL_\infty (t) - 2 \int dx dy |V (x-y)| |\ph_t (y)|^2 a_x^* a_x -\frac{1}{N} \int dx dy \, |V( x-y)| \, a_x^* a_y^* a_y a_x
\end{split} \]
The lower bound for $\cL (t)$ follows now by the lower bound for $\cL_\infty (t)$, by the bound (\ref{eq:cW}), and since clearly
\[ \int dx dy |V (x-y)| |\ph_t (y)|^2 a_x^* a_x \leq \| |V| * |\ph_t|^2 \|_\infty \, \cN \leq C \cN \]
The upper bound for $\cL (t)$ follows similarly.
\end{proof}

Finally, in the next lemma we show how to bound creation and annihilation operators by the number of particles operator after commuting them through the kinetic energy.
\begin{lemma}\label{lm:phiNK}
For $j \in \bN$, there exists a constant $C$, such that
\begin{equation}\label{eq:phiNK} \begin{split} a^* (\ph) (\cN+1)^j \cK (\cN+1)^j a (\ph) &\leq  \| \ph \|^2 \, \cN^{j+1/2} \, \cK \,  \cN^{j+1/2} \qquad \text{and } \\   a(\ph) (\cN+1)^j \cK (\cN +1)^j a^* (\ph)  &\leq C \, \| \ph \|^2_{H^1}  \, (\cN+1)^{j+1/2} (\cK + 1) (\cN+1)^{j+1/2} \end{split}\end{equation}
\end{lemma}
\begin{proof}
We observe that, if $\cK^{(n)}$ denotes the restriction of $\cK$ onto the $n$ particle sector,
\[ \begin{split}
\Big\langle \psi, a^* (\ph) (\cN+1)^j \cK (\cN+1)^j \, a(\ph) \Omega \Big\rangle = & \, \sum_{n \geq 0} \left\langle \left( (\cN+1)^j a(\ph) \psi\right)^{(n)} , \cK^{(n)} \left((\cN+1)^j a(\ph) \psi\right)^{(n)} \right\rangle \\ = & \, \sum_{n \geq 0} (n+1)^{2j} \sum_{j=1}^n \int dx_1 \dots dx_n \, \left| \nabla_{x_j} (a(\ph) \psi)^{(n)} (x_1, \dots , x_n) \right|^2 \end{split}\]
Since \[ (a(\ph) \psi)^{(n)} (x_1, \dots, x_n) = (n+1)^{1/2} \int dx \, \overline{\ph} (x) \, \psi^{(n+1)} (x,x_1, \dots , x_n) \] we conclude that
\[ \begin{split}
\Big\langle &\psi, a^* (\ph) (\cN+1)^j \cK (\cN+1)^j \, a(\ph) \Omega \Big\rangle \\ = \; &\sum_{n \geq 0} (n+1)^{2j+1} \sum_{j=1}^n \int dx_1 \dots dx_n dx dy \, \overline{\ph} (x) \, \ph (y) \, \nabla_{x_j} \psi^{(n+1)} (x,x_1, \dots ,x_n) \, \nabla_{x_j} \overline{\psi}^{(n+1)} (y, x_1, \dots , x_n) \\ \leq \; & \sum_{n\geq 0} (n+1)^{2j+1} \sum_{j=1}^n \int dx_1 \dots dx_n dx dy |\ph (x)|^2 \, |\nabla_{x_j} \psi^{(n+1)} (y,x_1, \dots , x_n)|^2
\\ \leq \; & \| \ph \|^2 \sum_{n\geq 0} (n+1)^{2j+1} \sum_{j=1}^{n+1} \int dx_1 \dots dx_n dx_{n+1} \, |\nabla_{x_j} \psi^{(n+1)} (x_1, \dots , x_n,x_{n+1})|^2
\\ \leq \; & \| \ph \|^2 \, \left\langle \psi, \cN^{j+1/2} \cK \cN^{j+1/2} \psi \right\rangle
\end{split} \]
To prove the second bound in (\ref{eq:phiNK}), we commute $a^* (\ph)$ to the left and $a(\ph)$ to the right of $\cK$. Since
\[ \begin{split} a(\ph) \cK a^* (\ph) = \; &a(\ph) a^* (\ph) \cK + a(\ph) a^* (-\Delta \ph) \\ = \;& a^* (\ph) a(\ph) \cK + \| \ph \|^2 \, \cK + a(\ph) a^* (-\Delta \ph) \\ = \; & a^* (\ph) \cK a(\ph) + \cK + a^* (\ph) a(-\Delta \ph) +  a(\ph) a^* (-\Delta \ph) \\ = \; & a^* (\ph) \cK a(\ph) + \| \ph \|^2 \cK + a^* (\ph) a(-\Delta \ph) +  a^* (-\Delta \ph) a(\ph) + \| \nabla \ph \|^2 \end{split} \]
we find
\[ \begin{split} a (\ph) (\cN+1)^j \cK &(\cN + 1)^j a^* (\ph)  \\ = &(\cN+2)^j a(\ph) \cK a^*(\ph) (\cN+2)^j \\ \leq \; & (\cN+2)^j a^* (\ph) \cK a(\ph) (\cN+2)^j + \| \ph \|^2 \, (\cN+2)^j \cK (\cN+2)^j \\ &+ (\cN+2)^j   \left(a^* (\ph) a(-\Delta \ph) + a^* (-\Delta \ph) a (\ph) \right) (\cN+2)^j + \| \nabla \ph \|^2 (\cN+2)^{2j} \\ \leq \; & a^* (\ph) (\cN+3)^j \cK (\cN+3)^j a(\ph) + (\cN+2)^j \left( \| \ph \|^2 \, \cK + \| \nabla \ph \|^2 \right) (\cN+2)^j \\ &+ (\cN+2)^j   \left(a^* (\ph) a(-\Delta \ph) + a^* (-\Delta \ph) a (\ph) \right) (\cN+2)^j \end{split} \]
Using the first bound in (\ref{eq:phiNK}) to control the first term, and the observation that
\[ \begin{split} \Big\langle \psi, (\cN+2)^j   &\left(a^* (\ph) a(-\Delta \ph) + a^* (-\Delta \ph) a (\ph) \right) (\cN+2)^j \psi \Big\rangle \\ \leq \; & 2 \| a(\ph) (\cN+2)^j \psi \| \, \| a(-\Delta \ph) (\cN+2)^j \psi \| \\ \leq  \; & \| \ph \|^2 \langle \psi, (\cN+2)^{2j+1} \psi \rangle + \| \nabla \ph \|^2 \langle \psi, (\cN+2)^j \cK (\cN+2)^j \psi \rangle \end{split}
\]
where in the last inequality we integrated by parts to move one of the derivatives in $-\Delta \ph$ back onto $\psi$. We conclude that
\[  a (\ph) (\cN+1)^j \cK (\cN + 1)^j a^* (\ph) \leq C \| \ph \|_{H^1}^2 \, (\cN+1)^{j+1/2} (\cK+1) (\cN+1)^{j+1/2} \]
\end{proof}

\section{Regularization of the potential}
\label{sec:reg}
\setcounter{equation}{0}

As in \cite{CLS}, our analysis requires a regularization of the interaction potential. This is needed to establish the convergence of the fluctuation dynamics $\cU (t;s)$ to its limit $\cU_\infty (t;s)$. We introduce therefore an $N$-dependent cutoff in $V$, vanishing sufficiently fast in the limit $N \to \infty$. The last condition should guarantee that the quantities we are interested in do not change when $V$ is replaced by the regularized potential. On the other hand, the presence of the small cutoff will give us an ($N$-dependent) $L^{\infty}$-bound on the potential which is crucial in the analysis (in particular, in Section~\ref{sec:comp}).

For an arbitrary sequence $\alpha_N > 0$, we set
\begin{equation}\label{eq:wtV}
\widetilde V (x) = \text{sgn}(V (x)) \cdot \min \{ |V (x)|, \alpha_N^{-1} \}
\end{equation}
where $\text{sgn} (V(x))$ denotes the sign of $V (x)$.
Note that, by definition $|\wt{V} (x)| \leq \alpha_N^{-1}$. Moreover, the assumption $V^2 (x) \leq D (1-\Delta)$ implies that the same bound holds also for $\wt{V}$, independently of $N$:
\begin{equation}\label{eq:wtVcond} \wt{V}^2 (x) \leq D (1-\Delta_x) \, . \end{equation}
We define the regularized Hamiltonian
\begin{equation}\label{eq:regham}
\widetilde{H}_N = \sum_{j=1}^N -\Delta_{x_j} + \frac{1}{N} \sum_{i<j}^N \widetilde{V}(x_i -x_j). \end{equation}
On the Fock space $\cF$, we define
\begin{equation}\label{eq:regham-FS}
\wt{\cH}_N = \int dx \nabla_x a^*_x \nabla_x a_x + \frac{1}{2N} \int dx dy \, \wt{V} (x-y) \, a^*_x a^*_y a_y a_x \end{equation}
and the fluctuation dynamics \begin{equation}\label{eq:U-reg} \wt{\cU} (t;s) = e^{-i\wt{\omega} (t;s)} \, W^* (\sqrt{N} \wt{\ph}_t) e^{-i \wt{\cH}_N (t-s)} W (\sqrt{N} \wt{\ph}_s) \end{equation}
with the phase
\[ \wt\omega(t;s) = \frac{N}{2} \int_s^t d\tau \int dx (\wt V*|\wt\ph_{\tau}|^2 ) (x) |\wt \ph_{\tau} (x)|^2 \]
and where $\wt{\ph}_t$ solves the regularized ($N$ dependent) Hartree equation \begin{equation}\label{eq:reg-hartree} i\partial_t \wt{\ph}_t = -\Delta \wt{\ph}_t + \left( \wt{V} * |\wt{\ph}_t|^2 \right) \wt{\ph}_t \end{equation}
We denote the generator of $\wt{\cU} (t;s)$ by $\wt{\cL} (t)$. Clearly, $\wt{\cL} (t)$ has the form (\ref{eq:L}), but with $V$ and $\ph_t$ replaced by $\wt{V}$ and, respectively, $\wt{\ph}_t$. We also modify the observable $\cO_t$ defined in (\ref{eq:obser}) by setting
\begin{equation}\label{eq:obser-reg}
\wt{\cO}_t := \frac{1}{\sqrt{N}} \sum_{j=1}^N \left( O^{(j)} - \E_{\wt{\ph}_t} O \right) = \frac{1}{\sqrt{N}} \sum_{j=1}^N \left( O^{(j)} - \langle \wt{\ph}_t ,  O \wt{\ph}_t \rangle \right)
\end{equation}

To compare the original many body evolution (generated by the Hamiltonian $H_N$) with the regularized many-body evolution (generated by $\wt{H}_N$), we use the following lemma.
\begin{lemma}\label{lm:compare}
Let $\psi_{N,t} = e^{-iH_N t} \ph^{\otimes N}$ and $\wt{\psi}_{N,t} = e^{-i \wt{H}_N t} \ph^{\otimes N}$, with $\ph \in H^1 (\bR^3)$ such that $\| \ph \| = 1$. Then there exists a universal constant $C>0$ such that
\begin{equation}
\left\| \psi_{N,t} - \wt{\psi}_{N,t} \right\|^2 \leq C N \alpha_N \, |t| \,
\end{equation}
for all $N \in \bN$, $t \in \bR$.
\end{lemma}
The proof of this lemma can be found in \cite{CLS}[Lemma 2.1]. We can also compare the solution of the original Hartree equation (\ref{eq:hartree}), with the solution of the regularized Hartree equation (\ref{eq:reg-hartree}).
\begin{lemma}\label{lm:comp-ph}
Let $\ph \in H^1 (\bR^3)$, with $\| \ph \|_2=1$. Let $\ph_t$ and $\wt{\ph}_t$ denote the solutions of the Hartree equations (\ref{eq:hartree}) and, respectively, (\ref{eq:reg-hartree}), with initial data $\ph$. Then there are constants $C,K >0$ such that
\[ \| \ph_t - \wt{\ph}_t \| \leq C \alpha_N  e^{K |t|} \]
for every $t\in \bR$.
\end{lemma}
The proof of this lemma can be found in \cite{CLS}[Lemma 2.2].

As a consequence of the last two lemmas, it is enough to show (\ref{eq:mom}) for the regularized quantities, assuming that the cutoff is sufficiently small. In other words, Lemma \ref{lm:mom} follows by the next lemma, where the original potential $V$ is replaced by its regularized version $\wt{V}$.
\begin{lemma} \label{lm:mom-reg}
Fix $k \in \bN$. Choose an integer $r > 3 + k/2$ and let $\alpha_N = N^{-r}$. Let $\wt{V}$, $\wt{H}_N$, $\wt{\ph}_t$ and $\wt{\cO}_t$ be defined as in (\ref{eq:wtV}), (\ref{eq:regham}), (\ref{eq:reg-hartree}), (\ref{eq:obser-reg}) respectively. Let $\wt{\psi}_{N,t} = e^{-i \wt{H}_N t} \ph^{\otimes N}$. Then
\begin{equation}\label{eq:mom-reg} \lim_{N \to \infty} \E_{\wt{\psi}_{N,t}} \wt{\cO}_t^k = \lim_{N \to \infty} \left\langle \wt{\psi}_{N,t} , \wt{\cO}_t^k \, \wt{\psi}_{N,t} \right\rangle = \left\{ \begin{array}{ll} 0 \quad &\text{if $k$ is odd} \\
\frac{k!}{2^{k/2} (k/2)!} \sigma_t^{2k} \quad &\text{if $k$ is even} \end{array} \right. \end{equation}
with $\sigma_t$ as defined in (\ref{eq:var}).
\end{lemma}

Using Lemma \ref{lm:mom-reg}, we can prove Lemma \ref{lm:mom} (and thus Theorem \ref{thm:CLT}) as follows.

\begin{proof}[Proof of Lemma \ref{lm:mom}]
Fix $k \in \bN$, $t\in \bR$. Choose $r > k/2 + 3$ and let $\alpha_N = N^{-r}$. Observe that
\[ \| \cO_t \| = \left\| \frac{1}{\sqrt{N}} \sum_{j=1}^N \left( O^{(j)} - \E_{\ph_t} O \right) \right\| \leq \sqrt{N} \, \| O - \E_{\ph_t} O \| \leq 2 \sqrt{N} \| O \| \]
The same bound holds for $\wt{\cO}_t$. Moreover, from Lemma \ref{lm:comp-ph}, we have
\[ \| \cO_t - \wt{\cO}_t \| \leq \sqrt{N} \left| \langle \ph_t , O \ph_t \rangle - \langle \wt{\ph}_t , O \wt{\ph}_t \rangle \right| \leq 2 \sqrt{N} \| \ph_t - \wt{\ph}_t \| \| O \| \leq C \sqrt{N} \alpha_N e^{K|t|} \]
This implies that
\[ \| \cO_t^k - \wt{\cO}_t^k \| \leq  C N^{k/2} \alpha_N e^{K|t|} \] for an appropriate $k$-dependent constant $C$.
Therefore, using Lemma \ref{lm:compare}, we find
\[ \begin{split} \Big|\E_{\psi_{N,t}} \cO^k_t - \E_{\wt{\psi}_{N,t}} \wt{\cO}^{k}_t \Big| = \; & \Big| \langle \psi_{N,t}, \cO_t^{k} \, \psi_{N,t} \rangle - \langle \wt{\psi}_{N,t} , \wt{\cO}_t^{k}\, \wt{\psi}_{N,t} \rangle \Big| \\ \leq \; & \Big| \langle \psi_{N,t}, \left( \cO_t^{k} - \wt{O}_t^{k} \right) \, \psi_{N,t} \rangle \Big| + \Big| \langle \psi_{N,t} , \wt{\cO}_t^{k} \, \psi_{N,t} \rangle
- \langle \wt{\psi}_{N,t} , \wt{\cO}_t^{k} \, \wt{\psi}_{N,t} \rangle \Big| \\ \leq \; & \left\| \cO^k_t - \wt{\cO}^k_t \right\| + 2 \left\| \psi_{N,t} - \wt{\psi}_{N,t} \right\| \, \| \wt{\cO}_t \|^{k} \\
\leq \; & C  \, N^{1+ k/2} \alpha_N  e^{K|t|} = C  N^{1+k/2-r} e^{K|t|} \end{split} \]
which converges to zero, as $N \to \infty$.
We conclude that
\[ \lim_{N \to \infty} \E_{\psi_{N,t}} \cO^k_t = \lim_{N \to \infty} \E_{\wt{\psi}_{N,t}} \wt{\cO}^k_t \]
and therefore (\ref{eq:mom}) follows from Lemma \ref{lm:mom-reg}.
\end{proof}

\section{Comparison of Dynamics}
\label{sec:comp}
\setcounter{equation}{0}

Finally, we will need to compare the (regularized) fluctuation dynamics $\wt{\cU} (t;s)$ with the limiting dynamics $\cU_\infty (t;s)$ defined in (\ref{eq:Uinfty}) (note that the limiting dynamics is defined in terms of the original potential $V$, and the solution of the Hartree equation (\ref{eq:hartree}), without cutoff; the reason is that, in the limit $N \to \infty$, the $N$-dependent cutoff $\alpha_N$ disappears).

\begin{proposition}\label{prop:3}
Suppose that, in the definition (\ref{eq:wtV}) of the regularized potential $\wt V$, the cutoff $\alpha_N$ is such that $\alpha_N \geq N^{-r}$, for some $r \in \bN$, $r \geq 1$. Then, for any $j \in \bN/2$, there exist constants $C, K >0$ depending on $r$ and $j$, such that
\begin{equation}\label{eq:claim-comp}
\begin{split}
\Big\| (\cN&+1)^j \left( \wt{\cU} (t;s) - \cU_\infty (t;s) \right) \psi \Big\| \\ & \leq C e^{K |t-s|} \left( N^{-1/2} + \alpha_N (N+1)^{j+1} \right) \, \left\| \left( \cK +  \cN^{(4r-2) (4+2j) + 6} + 1 \right)^{1/2} (\cN+1)^{2j+2} \psi \right\|
\end{split} \end{equation}
\end{proposition}

Note that the bound (\ref{eq:claim-comp}) is only useful for $j < r-1$; it will be applied under this condition.

\begin{proof}
We fix $t \geq 0$ and $s =0$ (all other cases can be treated analogously). We use
\begin{equation}\label{eq:comp1}  \wt{\cU} (t;0) - \cU_\infty (t;0) = \int_0^t d\tau \, \cU_\infty (t;\tau) \left( \wt{\cL} (\tau) - \cL_\infty (\tau) \right) \wt{\cU} (\tau;0) \end{equation}
where $\wt{\cL} (t)$ is the generator of the regularized fluctuation dynamics $\wt{\cU} (t;s)$ and has the form (\ref{eq:L}), with $V$ and $\ph_t$ replaced by $\wt{V}$ and $\wt{\ph}_t$.
Hence
\begin{equation} \label{eq:dec-cL}
\begin{split}
\left\| (\cN+1)^j \left( \wt{\cU} (t;0) - \cU_\infty (t;0) \right) \psi \right\| \leq \; & \int_0^t d \tau \, \left\| (\cN+1)^j \cU_\infty (t;\tau) \cL_2 (\tau) \wt{\cU} (\tau;0) \psi \right\| \\ &+ \int_0^t d \tau \, \left\| (\cN+1)^j \cU_\infty (t;\tau) \cL_3 (\tau) \wt{\cU} (\tau;0) \psi \right\| \\ &+ \int_0^t d \tau \, \left\| (\cN+1)^j \cU_\infty (t;\tau) \cL_4 \wt{\cU} (\tau;0) \psi \right\|
\end{split}\end{equation}
where we decomposed
\[ \wt{\cL} (\tau) - \cL_\infty (\tau) = \cL_2 (\tau) + \cL_3 (\tau) + \cL_4 \]
with
\[ \begin{split} \cL_2 (\tau) = & \; \int dx \left(  \wt{V}* |\wt{\ph}_t|^2 (x) - V*|\ph_t|^2 (x) \right) a_x^* a_x \\ &+ \int dx dy \, \left( \wt{V} (x-y) \wt{\ph}_t (x) \overline{\wt{\ph}} (y) -  V(x-y) \ph_t (x) \overline{\ph}_t (y) \right) a_x^* a_y \\ & + \int dx dy \left( \wt{V} (x-y) \wt{\ph}_t (x) \wt{\ph}_t (y) - V(x-y) \ph_t (x) \ph_t (y) \right) a_x^* a_y^* + \text{h.c} \end{split} \]
and
\[ \begin{split}
\cL_3 (\tau ) &= \frac{1}{\sqrt{N}} \int dx dy \, \wt{V} (x-y) \, a_x^* \left( \wt{\ph}_t (y) a_y^* + \overline{\wt{\ph}}_t (y) a_y \right) a_x \\
\cL_4 &=\frac{1}{N} \int dx dy \, \wt{V} (x-y) \, a_x^* a_y^* a_y a_x
\end{split} \]
We begin by estimating the first term on the r.h.s. of (\ref{eq:dec-cL}).
{F}rom Proposition \ref{prop:2} we have
\begin{equation}\label{eq:comp2} \left\| (\cN+1)^j  \cU_\infty (t;\tau) \cL_2 (\tau) \wt{\cU} (\tau;0) \psi \right\| \leq C e^{K (t-\tau)} \left\| (\cN+1)^{j} \cL_2 (\tau) \wt{\cU} (\tau;0) \psi \right\| \end{equation}
We write
\[ \cL_2 (\tau) = A_{1,\tau} + A_{2,\tau} + A_{3,\tau} \]
where
\[ \begin{split}  A_{1,\tau} &= \int dx \left(  \wt{V}* |\wt{\ph}_t|^2 (x) - V*|\ph_t|^2 (x) \right) a_x^* a_x \\ A_{2,\tau}  &= \int dx dy \, \left( \wt{V} (x-y) \wt{\ph}_t (x) \wt{\ph} (y) -  V(x-y) \ph_t (x) \overline{\ph}_t (y) \right) a_x^* a_y \\ A_{3,\tau} &= B_\tau +B_\tau^* \end{split} \]
with
\[ B_\tau =  \int dx dy \left( \wt{V} (x-y) \wt{\ph}_t (x) \wt{\ph}_t (y) - V(x-y) \ph_t (x) \ph_t (y) \right) a_x^* a_y^*  \]
{F}rom (\ref{eq:comp2}), we find
\begin{equation}\label{eq:comp3} \begin{split}
\Big\| (\cN+1)^j \cU_\infty &(t;\tau) \cL_2 (\tau) \wt{\cU} (\tau;0) \psi \Big\| \\ \leq \; &C e^{K (t-\tau)} \left\| A_{1,\tau} \, (\cN+1)^{j} \wt{\cU} (\tau;0) \psi \right\| +  C e^{K (t-\tau)} \left\| A_{2,\tau} \, (\cN+1)^{j} \wt{\cU} (\tau;0) \psi \right\| \\ &+ C e^{K (t-\tau)} \left\| B_{\tau} \, (\cN+3)^{j} \wt{\cU} (\tau;0) \psi \right\| + C e^{K (t-\tau)} \left\| B^*_{\tau} \, (\cN-1)^{j} \wt{\cU} (\tau;0) \psi \right\| \end{split} \end{equation}
Using $|V - \wt{V}| \leq \alpha_N |V|^2$, the assumption $V^2 \leq D (1-\Delta)$, the propagation of regularity bound (\ref{eq:phHs}) and the bound $\| \ph_t - \wt{\ph}_t \|_2 \leq C \alpha_N e^{K|t|}$ from Lemma \ref{lm:comp-ph}, we find a constant $C$ such that
\[
\sup_{x} \, \left| (\wt{V}* |\wt{\ph}_\tau|^2) (x) - (V*|\ph_\tau|^2) (x) \right|  \leq \; C \alpha_N e^{K \tau}
\]
This bounds implies that
\[ A_{1,\tau}^2 \leq C \alpha_N^2 e^{K \tau} \, \cN^2 \]
Similarly, we find $A_{2,\tau}^2 \leq C \alpha_N^2 e^{K\tau} \, \cN^2$.
To bound the last two terms on the r.h.s. of (\ref{eq:comp3}), we write
\[ \begin{split} B_\tau = \; & \int dx dy \, \left( \wt{V} (x-y) - V (x-y) \right) \wt{\ph}_t (x) \wt{\ph}_t (y) \, a_x^* a_y^* \\ &+ \int dx dy \, V (x-y) \left(\wt{\ph}_t (x) - \ph_t (x) \right) \wt{\ph}_t (y) \, a_x^* a_y^* +\int dx dy \, V (x-y) \ph_t (x) \left(\wt{\ph}_t (y) - \ph_t (y) \right) \, a_x^* a_y^* \end{split} \]
Cauchy-Schwarz inequality implies
\begin{equation}\label{eq:BIII} \begin{split} \left\langle \psi , B_{\tau} B^*_{\tau}\psi \right\rangle  \leq \; & \int dx dy dz dw \, \left(\wt{V} (x-y)- V (x-y)\right) \left( \wt{V}(z-w) - V (z-w) \right) \\ &\hspace{1cm} \times \wt{\ph}_t (x) \wt{\ph}_t (y) \overline{\wt{\ph}}_t (z) \overline{\wt{\ph}}_t (w)
\left\langle \psi, a_x^* a_y^* a_z a_w \psi \right\rangle   \\ &+ \int dx dy dz dw \, V(x-y)V(z-w) (\wt{\ph}_t (x) - \ph_t (x)) (\overline{\wt{\ph}}_t (z) - \overline{\ph}_t (z)) \\& \hspace{1cm} \times  \wt{\ph}_t (y) \overline{\wt{\ph}}_t (w)
\left\langle \psi, a_x^* a_y^* a_z a_w \psi \right\rangle \\ &+ \int dx dy dz dw \, V(x-y)V(z-w) (\wt{\ph}_t (y) - \ph_t (y)) (\overline{\wt{\ph}}_t (w) - \overline{\ph}_t (w)) \\ &\hspace{1cm} \times \ph_t (x) \overline{\ph}_t (z)
\left\langle \psi, a_x^* a_y^* a_z a_w \psi \right\rangle \\  = \; &\text{I} + \text{II} + \text{III} \end{split} \end{equation}
Using $|V-\wt{V}| \leq \alpha_N V^2$ and the operator inequality $V^2 \leq D (1-\Delta)$, the first term is bounded by
\[ \begin{split}
\text{I} \leq \; & \alpha_N^2 \int dx dy dz dw V^2 (x-y) \, V^2 (z-w) |\wt{\ph}_t (x)| |\wt{\ph}_t (y)| |\wt{\ph}_t (z)| |\wt{\ph}_t (w)| \, \| a_x a_y \psi \| \, \| a_z a_w \psi \| \\ \leq \; & \alpha_N^2 \int dx dy dz dw V^2 (x-y) \, V^2 (z-w) \, |\wt{\ph}_t (z)|^2 |\wt{\ph}_t (w)|^2 \, \| a_x a_y \psi \|^2 \\ \leq \; &C \alpha_N^2 \int dxdy \, \left( \| \nabla_x a_x a_y \psi \|^2 + \| a_x a_y \psi \|^2 \right) \\ \leq \; &C \alpha_N^2 \left\langle \psi, \left( \cK \cN + \cN^2 \right) \psi \right\rangle
\end{split} \]
The second term on the r.h.s. of (\ref{eq:BIII}), on the other hand, can be estimated by
\[ \begin{split}
\text{II} \leq \; & \int dx dy dz dw \, V^2 (x-y) \, |\ph_t (y)|^2 | \ph_t (x) - \wt{\ph}_t (x)|^2 \, \| a_z a_w \psi \|^2 \\ \leq \; &  \| V^2 * |\ph_t|^2 \|_\infty \, \| \ph_t - \wt{\ph}_t \|_2^2  \, \langle \psi, \cN^2 \psi \rangle \\ \leq \; & C \alpha_N^2 e^{Kt} \langle \psi, \cN^2 \psi \rangle
\end{split} \]
where we used Lemma \ref{lm:comp-ph}. The third term on the r.h.s. of (\ref{eq:BIII}) can be handled similarly. Hence
\[ \langle \psi, B_\tau B_\tau^* \psi \rangle \leq C \alpha_N^2 e^{K \tau} \langle \psi, \left( \cK \cN + \cN^2 \right) \psi \rangle \]
Analogously, one can prove that
\[ \langle \psi, B^*_\tau B_\tau \psi \rangle \leq C \alpha_N^2 e^{K \tau} \langle \psi, \left( \cK \cN + \cN^2 + 1 \right) \psi \rangle \]
(the factor $1$ takes into account the contribution of the commutator between $B_\tau$ and $B_\tau^*$). We conclude from (\ref{eq:comp3}) that
\begin{equation}\label{eq:L2-fin} \begin{split}
\left\| (\cN+1)^j \cU_\infty (t;\tau) \cL_2 (\tau) \wt{\cU} (\tau;0) \psi \right\| \leq \; &C \alpha_N \, e^{K t} \left\| \, (\cN+1)^{j+1} \wt{\cU} (\tau;0) \psi \right\| \\ &+ C \alpha_N \, e^{K t} \left\| \cK^{1/2} (\cN+1)^{j+1/2} \wt{\cU} (\tau;0) \psi \right\| \\
\leq \; &C \alpha_N \, e^{K t} \left\| \, (\cN+1)^{2j+3} \psi \right\| \\ &+C \alpha_N \, (N+1)^{\lceil j+1/2 \rceil} e^{Kt} \, \left\| \cK^{1/2} \wt{\cU} (\tau;0) \psi \right\| \\ &+ C \alpha_N e^{Kt} \left\| (\cK + \cN^2 + 1)^{1/2} \, ( \cN + N+1)^{\lceil j+1/2 \rceil} \psi \right\| \\
\leq \; &C \alpha_N \, (N+1)^{j+1} e^{Kt} \, \left\| (\cK + \cN^2 + 1)^{1/2} \, ( \cN + 1)^{2j+2} \psi \right\|
\end{split} \end{equation}
where we applied Proposition \ref{prop:1} and Lemma \ref{lm:NKN} (and where $\lceil j+1/2 \rceil$ denotes the smallest integer larger or equal to $j+1/2$).

To bound the second term on the r.h.s. of (\ref{eq:dec-cL}), we use the inequality
\[ \cL_3 (t) \, (\cN +1)^{2j} \cL_3 (t) \leq \frac{C}{N} (\cN+1)^{2j+3} \]
which is proven in \cite{CLS}[Lemma 6.3]. This implies that
\begin{equation}\label{eq:L3-fin} \begin{split}
\left\| (\cN+1)^j \cU_\infty (t;\tau) \cL_3 (\tau) \wt{\cU} (\tau;0) \psi \right\| \leq \; & C e^{K(t-\tau)} \left\| (\cN+1)^j \cL_3 (\tau) \wt{\cU} (\tau;0) \psi \right\| \\ \leq \; & C N^{-1/2} e^{K(t-\tau)} \left\| (\cN+1)^{j+3/2} \wt{\cU} (\tau;0) \psi \right\| \\ \leq \; &C N^{-1/2} e^{Kt} \, \| (\cN+1)^{2j+4} \psi \| \end{split} \end{equation}

Finally, to bound the last term on the r.h.s. of (\ref{eq:dec-cL}), we note that, assuming $\alpha_N \geq N^{-r}$,
\begin{equation}\label{eq:NL4N} (\cN+1)^j \cL_4^2 (\cN+1)^j \leq C N^{-1} \left( \cK + \cN^{2r (2j+4)} +1 \right) \end{equation}
This is in fact the only place in the proof where a cutoff in the potential is needed. To prove (\ref{eq:NL4N}), we notice that, when restricted to the $n$-particle sector
\[\begin{split} \cL_4^2 (\cN+1)^{2j} \downharpoonright_{\cF^{(n)}} = \; & \frac{(n+1)^{2j}}{N^2} \left( \sum_{i<j}^n \wt{V} (x_i -x_j) \right)^2 \leq \frac{(n+1)^{2j+2}}{N^2} \sum_{i<j}^n \wt{V}^2 (x_i -x_j) \\
\leq \; &\frac{(n+1)^{2j+3} \, {\bf 1} \left( n \leq N^{\frac{1}{2j+3}} \right)}{N^2} \sum_{j=1}^n (1-\Delta_{x_j})+ \frac{(n+1)^{2j+4}\, {\bf 1} \left( n > N^{\frac{1}{2j+3}} \right)}{\alpha_N^2 N^2} \\ \leq \; & \frac{1}{N} \left( \cK + \cN \right) \downharpoonright_{\cF^{(n)}} + \frac{ (\cN+1)^{2j+4} \, {\bf 1} (\cN > N^{\frac{1}{2j+3}})}{N^2 \alpha_N^2} \downharpoonright_{\cF^{(n)}}
\end{split} \]
Assuming $\alpha_N \geq N^{-r}$, we find
\[ \frac{(\cN+1)^{2j+4} {\bf 1} (\cN > N^{\frac{1}{2j+3}})}{N^2 \alpha^2_N} \leq N^{2r-2- \frac{\ell}{2j+3}} \, (\cN+1)^{2j+4+\ell} \]
for any $\ell  \geq 0$. Choosing $\ell = (2r-1) (2j+3)$, we get (\ref{eq:NL4N}). This implies that the third term on the r.h.s. of (\ref{eq:dec-cL}) is bounded by
\begin{equation}\label{eq:L4-fin} \begin{split} \| (\cN+1)^j \cU_\infty (t,\tau)\, \cL_4 \wt{\cU} (t;0) \psi \| \leq \; & C N^{-1/2} e^{K (t-\tau)} \left\| \left( \cK + \cN^{2r (2j+4)} +1 \right)^{1/2} \, \wt{\cU} (\tau,0) \psi \right\| \\ \leq \; & C N^{-1/2} \, e^{Kt} \, \left\| \left( \cK +  \cN^{4r (4+2j) + 2} + 1 \right)^{1/2} \psi \right\|
\end{split} \end{equation}
where we used again Proposition \ref{prop:1}.
Inserting (\ref{eq:L2-fin}), (\ref{eq:L3-fin}), (\ref{eq:L4-fin}) in (\ref{eq:dec-cL}), we obtain (\ref{eq:claim-comp}).
\end{proof}

\section{Computation of the moments}
\label{sec:mom}

In this section we prove Lemma \ref{lm:mom-reg}. To this end, we fix $k \in \bN$ and we choose an integer $r > k/2 +3$. We let $\alpha_N = N^{-r}$, and we define the regularized potential $\wt{V}$ as in (\ref{eq:wtV}) in terms of $\alpha_N$. For the given compact operator $O$ on $L^2 (\bR^3)$ we introduce the notation $\wt{O}_c = O - \langle \wt{\ph}_t , O \wt{\ph}_t \rangle$. Later, we will also use the notation $O_c = O - \langle \ph_t , O \ph_t \rangle$ (as a rule, notations with a tilde denote quantities defined in terms of the regularized potential $\wt{V}$ and the solution $\wt{\ph}_t$ of the corresponding Hartree equation (\ref{eq:reg-hartree})). We write
\[ \begin{split} \E_{\wt{\psi}_{N,t}} \wt{\cO}_t^k = \; &\frac{1}{N^{k/2}} \left\langle \wt{\psi}_{N,t} , \left( \sum_{j=1}^N \wt{O}_c^{(j)} \right)^{k} \, \wt{\psi}_{N,t} \right\rangle \\ = \; &\frac{1}{N^{k/2}} \sum_{m=1}^{k} {N \choose m} \sum_{i_1 + \dots + i_m = k} \frac{k!}{i_1 ! \dots i_m!} \left\langle \wt{\psi}_{N,t} , \left(\wt{O}_c^{i_1} \otimes \dots \otimes \wt{O}_c^{i_m} \otimes 1^{(N-m)} \right) \wt{\psi}_{N,t} \right\rangle
\end{split} \]
The sum over $i_1, \dots , i_m$ runs over all integers $i_1, \dots , i_m \geq 1$ with $i_1 + \dots  + i_m = k$. We separate next now the indices $i_1, \dots , i_m$ which are equal to one ($r$ denotes the number of such indices). We write
\[
\begin{split}
\E_{\wt{\psi}_{N,t}} \wt{\cO}_t^{k}
= \; &\frac{1}{N^{k/2}} \sum_{m=1}^{k} {N \choose m} \sum_{r=1}^m {m \choose r} \sum_{j_1 + \dots + j_{m-r} = k-r} \frac{k!}{j_1 ! \dots j_{m-r}!} \\ &\hspace{3cm} \times \left\langle \wt{\psi}_{N,t} , \left(\wt{O}_c \otimes \dots \otimes \wt{O}_c \otimes \wt{O}_c^{j_1} \otimes \dots \otimes \wt{O}_c^{j_{m-r}} \otimes 1^{(N-m)} \right) \wt{\psi}_{N,t} \right\rangle
\end{split} \]
where the sum runs over $j_1, \dots , j_{m-r}$ runs over all integers $j_1, \dots , j_{m-r} \geq 2$ with $j_1 + \dots + j_{m-r} = k-r$. In the following, we denote by $\wt{\psi}_{N,t}$ the vector $\{ 0, \dots , 0 , \wt{\psi}_{N,t} , 0, \dots \}$ in the Fock space $\cF$. Then, if $\wt{O}_c^i (x;y)$ denotes the integral kernel of the operator $\wt{O}_c^i$, we obtain
\[ \begin{split}
\E_{\wt{\psi}_{N,t}} \wt{\cO}_t^{k} = \; &\frac{1}{N^{k/2}} \sum_{m=1}^{k} \frac{1}{N^m} {N \choose m} \sum_{r=1}^m {m \choose r} \sum_{j_1 + \dots + j_{m-r} = k-r} \frac{k!}{j_1 ! \dots j_{m-r}!} \\ &\hspace{.5cm} \times \int dx_1 dx'_1 \dots dx_m dx'_m \,
\wt{O}_c (x_1 ; x'_1) \dots \wt{O}_c (x_r; x'_r) \wt{O}_c^{j_1} (x_{r+1} ; x'_{r+1}) \dots \wt{O}_c^{j_{m-r}} (x_m ; x'_m) \\ &\hspace{3cm}  \times \left\langle \frac{a^* (\ph)^N}{\sqrt{N!}} \Omega ,
 e^{i\wt{\cH}_N t} \, a^*_{x_1} \dots a^*_{x_m} a_{x'_m} \dots a_{x_1} \, e^{-i\wt{\cH}_N t} \frac{a^* (\ph)^N}{\sqrt{N!}} \Omega \right\rangle
\end{split} \]
with the (regularized) Hamiltonian $\wt{\cH}_N$ defined in (\ref{eq:regham-FS}). Next, we notice that
\[ \frac{a^* (\ph)^N}{\sqrt{N!}} \Omega = d_N P_N W (\sqrt{N} \ph) \Omega \] where $W(\sqrt{N} \ph)$ is the Weyl operator defined in (\ref{eq:weyl}), $P_N$ is the orthogonal projection onto the $N$-particle sector of the Fock space $\cF$, and $d_N = \sqrt{N!} e^{N/2} N^{-N/2} \simeq N^{1/4}$.
Denoting by $\wt{\cU} (t;0)$ the (regularized) fluctuation dynamics introduced in (\ref{eq:U-reg}), we find
\begin{equation}\label{eq:mome} \begin{split}
\E_{\wt{\psi}_{N,t}} \wt{\cO}_t^{k} = \; &\frac{1}{N^{k/2}} \sum_{m=1}^{k} \frac{1}{N^m} {N \choose m} \sum_{r=1}^m {m \choose r} \sum_{j_1 + \dots + j_{m-r} = k-r} \frac{k!}{j_1 ! \dots j_{m-r}!}  \\ &\hspace{.5cm} \times \int dx_1 dx'_1 \dots dx_m dx'_m \,
\wt{O}_c (x_1 ; x'_1) \dots \wt{O}_c (x_r; x'_r) \wt{O}_c^{j_1} (x_{r+1} ; x'_{r+1}) \dots \wt{O}_c^{j_{m-r}} (x_m ; x'_m) \\ &\hspace{1cm} \times  \left\langle \xi_N , \wt{\cU}^* (t;0) \, (a^*_{x_1} + \sqrt{N} \overline{\wt{\ph}}_t (x_1)) \dots (a^*_{x_m} + \sqrt{N} \overline{\wt{\ph}}_t (x_m))\right. \\ & \hspace{3cm} \left. \times (a_{x'_m} + \sqrt{N}\wt{\ph}_t (x'_m)) \dots (a_{x_1} + \sqrt{N} \wt{\ph}_t (x_1)) \, \wt{\cU} (t;0) \Omega \right\rangle
\end{split}
\end{equation}
where we introduced the notation
\begin{equation}\label{eq:xi-def} \xi_N = d_N W^* (\sqrt{N} \ph) \frac{a^* (\ph)^N}{\sqrt{N!}} \Omega \end{equation}
Note that, from \cite[Lemma 7.1]{CLS}, we have the estimate
\begin{equation}\label{eq:Nxi} \| (\cN+1)^{-1/2} \xi_N \| \leq C \end{equation}
for a constant $C>0$, uniformly in $N$. When we expand the product of the $2m$ parenthesis on the r.h.s. of (\ref{eq:mome}), we find several terms, with different numbers of factors $\wt{\ph}_t$ or $\overline{\wt{\ph}}_t$.
Each term has the form (after appropriate change of variables, and with a constant $C$ depending on $k,m,r,j_1, \dots , j_{m-r}$)
\begin{equation}\label{eq:typ-term} \begin{split}
C \, &N^{\frac{a+b+2d-k}{2}-m} \, {N \choose m} \langle \ph_t, \wt{O}_c^{r_1} \ph_t \rangle \dots \langle \ph_t, \wt{O}_c^{r_d} \ph_t \rangle
\int dx_1 dx'_1 \dots dx_n dx'_n \, \wt{O}_c^{p_1} (x_1, x'_1) \dots \wt{O}_c^{p_n} (x_n , x'_n) \\ &\times \left\langle \wt{\cU} (t;0) \xi_N , a^* (\wt{O}_c^{q_1} \ph_t) \dots a^* (\wt{O}_c^{q_a} \ph_t) a_{x_1}^* \dots a_{x_n}^* a_{x'_n} \dots a_{x'_1} \,
a(\wt{O}_c^{s_1} \ph_t) \dots a(\wt{O}_c^{s_b} \ph_t) \wt{\cU} (t;0) \Omega
\right\rangle
\end{split} \end{equation}
where $n,a,b,d$ are integers with $n+a+b+d = m$, and $\{p_\ell\}_{\ell=1}^n$, $\{ q_\ell \}_{\ell=1}^a$,$\{ r_\ell \}_{\ell=1}^d$, $\{ s_\ell \}_{\ell=1}^b$ are appropriate sequences of integers ($r$ of them are equal to 1). Note that $a+b+2d$ is the total number of factors $\ph_t$ and $\overline{\ph}_t$; since each such term carries a weight $N^{1/2}$, this explains the appearance of the factor $N^{(a+b+2d)/2}$). The absolute value of (\ref{eq:typ-term}) is bounded above by
\[ \begin{split}
C &N^{\frac{a+b+2d-k}{2}} \| \wt{O}_c \|^{r_1 + \dots + r_d} \int dx'_1 \dots dx'_n  \Big| \Big\langle \wt{\cU} (t;0) \xi_N , a^* (\wt{O}_c^{q_1} \ph_t) \dots a^* (\wt{O}_c^{q_a} \ph_t) \\ & \hspace{2cm} \times a^* (\wt{O}_c^{p_1} (.,x'_1)) \dots a^* (\wt{O}_c^{p_n} (.,x'_n))  a_{x'_n} \dots a_{x'_1} \, a(\wt{O}_c^{s_1} \ph_t) \dots a(\wt{O}_c^{s_b} \ph_t) \wt{\cU} (t;0) \Omega \Big\rangle \Big| \\ \leq \; & C N^{\frac{a+b+2d-k}{2}} \| \wt{O}_c \|^{r_1 + \dots + r_d} \int dx'_1 \dots dx'_n \\ &\hspace{.5cm} \times \left\| (\cN+1)^{2+(n+a)/2} \, a_{x'_n} \dots a_{x'_1} \,
a(\wt{O}_c^{s_1} \ph_t) \dots a(\wt{O}_c^{s_b} \ph_t) \wt{\cU} (t;0) \Omega
\right\| \\ &\hspace{.5cm} \times \left\| (\cN+1)^{-2-(n+a)/2} a (\wt{O}_c^{p_n} (.,x'_n)) \dots a (\wt{O}_c^{p_1} (.,x'_1)) a  (\wt{O}_c^{q_a} \ph_t) \dots  a (\wt{O}_c^{q_1} \ph_t) \wt{\cU} (t;0) \xi_N \right\| \\ \leq \; &  C N^{\frac{a+b+2d-k}{2}} \| \wt{O}_c \|^{2(s_1 + \dots + s_b) + 2 (r_1 + \dots + r_d)} \, \| (\cN+1)^{2+ n +(a+b)/2} \wt{\cU} (t;0) \Omega \| \\ &+
C N^{\frac{a+b+2d-k}{2}} \| \wt{O}_c \|^{2(q_1 + \dots + q_a)} \int dx'_1 \dots dx'_n  \| \wt{O}_c^{p_n} (.,x'_n) \|^2 \, \dots \| \wt{O}_c^{p_1} (.,x'_1) \|^2  \| (\cN+1)^{-2} \wt{\cU} (t;0) \xi_N \|^2
\\ \leq \; &
C \, N^{\frac{a+b+2d-k}{2}} \left( \| \wt{O}_c \|^{2(s_1 + \dots + s_b+r_1 + \dots + r_d)} + \| \wt{O}_c \|^{2(q_1 + \dots + q_a)} \| \wt{O}_c^{p_n}\|_{\text{HS}}^2 \, \dots \| \wt{O}_c^{p_1}\|_{\text{HS}}^2 \| (\cN+1)^{-1/2} \xi_N \|^2  \right)
\end{split} \]
where in the last line, we used Proposition \ref{prop:1} (which implies that $\| (\cN+1)^{-2} \wt{\cU} (t;0) (\cN+1)^{1/2} \| < C$). {F}rom (\ref{eq:Nxi}), the absolute value of (\ref{eq:typ-term}) is thus controlled by
\[ C \, N^{\frac{a+b+2d-k}{2}} \left( \| O \|^{2(s_1 + \dots + s_b+r_1 + \dots + r_d)} + \| O \|^{2(q_1 + \dots + q_a+p_1 + \dots +p_n-n)} \, \| O \|_{\text{HS}}^{2n} \right) \]
Hence, we conclude that only the terms with $k\leq a+b+2d$ give an important contribution in the limit $N \to \infty$. This implies that $m \geq k/2$ (and $m \geq (k+1)/2$, if $k$ is odd). Moreover, terms where both $\overline{\wt{\ph}}_t (x_k)$ and $\wt{\ph}_t (x_k)$ appear, for some $k =1,\dots ,r$, vanish because $\langle \wt{\ph}_t , \wt{O}_c \, \wt{\ph}_t \rangle =0$ (this is the reason for separating the indices which are equal to one). These two observations imply that only terms with $r \leq 2m - k$ may contribute in the limit $N \to \infty$ (if $r \geq 2m - k + 1$, terms with at least $k$ factors of $\wt{\ph}_t$ and $\overline{\wt{\ph}}_t$ will contain at least one factor $\langle \wt{\ph}_t, \wt{O}_c \wt{\ph}_t \rangle$ and therefore will vanish). On the other hand, since $j_1, \dots , j_{m-r} \geq 2$, we have
\[ k - r = j_1 + \dots + j_{m-r} \geq 2 (m-r) \]
which implies that $r \geq 2m - k$. We conclude that only terms with $r = 2m - k$, with $j_1 = \dots = j_{m-r} = 2$, and with exactly $k$ factors of $\wt{\ph}_t$ can contribute in the limit $N \to \infty$. More precisely, after integration, we obtain
\begin{equation}\label{eq:higher2} \begin{split}
\E_{\wt{\psi}_{N,t}} \wt{\cO}_t^k = \; & \sum_{m= k/2}^{k} \frac{1}{m!} {m \choose 2m - k} \frac{k!}{2^{k-m}}  \,\langle \wt{\ph}_t , \wt{O}_c^2 \wt{\ph}_t \rangle^{k -m} \\ &\hspace{3cm} \times  \langle \xi_N , \wt{\cU}^* (t;0) : \left( a (\wt{O}_c \wt{\ph}_t) + a^* (\wt{O}_c \wt{\ph}_t) \right)^{2m-k} : \wt{\cU} (t;0) \Omega \rangle \\ &+ o(1)
\end{split}
\end{equation}
as $N \to \infty$. Here $: \cdot :$ denotes the operation of normal ordering of the creation and annihilation operators; all creation operators have to be written on the left, and all annihilation operators have to be written on the right, as if they commuted (for example, $:a (f_1) a^* (f_2) a^* (f_3)  a(f_4):  \; = a^* (f_2) a^* (f_3) a (f_1) a (f_4)$).
Using Proposition \ref{prop:3} and (\ref{eq:Nxi}), we can replace, in (\ref{eq:higher2}), the (regularized) fluctuation dynamics $\wt{\cU} (t;0)$ with its limit $\cU_{\infty} (t;0)$, defined in (\ref{eq:Uinfty}). In fact,
\begin{equation}\label{eq:compa3}\begin{split}
\Big| &\left\langle
\xi_N , \wt{\cU}^* (t;0) : \left( a (\wt{O}_c \wt{\ph}_t) + a^* (\wt{O}_c \wt{\ph}_t) \right)^{2m-k} : \wt{\cU} (t;0) \Omega \right\rangle \\ & \hspace{4cm} - \left\langle
\xi_N , \cU_{\infty}^* (t;0) : \left( a (\wt{O}_c \wt{\ph}_t) + a^* (\wt{O}_c \wt{\ph}_t) \right)^{2m-k} : \cU_{\infty} (t;0) \Omega \right\rangle \Big| \\ & \hspace{2cm} \leq \left| \left\langle
\xi_N , \left(\wt{\cU}^* (t;0) - \cU_\infty^* (t;0) \right) : \left( a (\wt{O}_c \wt{\ph}_t) + a^* (\wt{O}_c \wt{\ph}_t) \right)^{2m-k} : \wt{\cU}_\infty (t;0) \Omega \right\rangle \right| \\ & \hspace{2.5cm} + \left|\left\langle \xi_N , \cU (t;0) : \left( a (\wt{O}_c \wt{\ph}_t) + a^* (\wt{O}_c \wt{\ph}_t) \right)^{2m-k} : \left(\wt{\cU} (t;0) -\cU_\infty (t;0) \right) \Omega \right\rangle \right|
\\ & \hspace{2cm} \leq  \| (\cN+1)^{-1/2} \xi_N \| \sum_{j=1}^{2m+k} {2m-k \choose j}  \\ &\hspace{2.5cm} \times \left[ \left\| (\cN+1)^{1/2} \left( \wt{\cU}^* (t;0) - \cU_\infty^* (t;0) \right) a^* (\wt{O}_c \wt{\ph}_t)^{j} a (\wt{O}_c \wt{\ph}_t)^{2m-k-j} \cU_\infty (t;0) \Omega \right\| \right. \\ &\hspace{3cm} \left. + \left\| (\cN+1)^{1/2} \cU (t;0) \, a^* (\wt{O}_c \wt{\ph}_t)^j  a (\wt{O}_c \wt{\ph}_t)^{2m-k-j} \left(\wt{\cU} (t;0) -\cU_\infty (t;0) \right) \Omega \right\| \right]
\end{split} \end{equation}
To bound the first term in the parenthesis, we use Proposition \ref{prop:3}. We find
\[ \begin{split}
\Big\| (\cN+1)^{1/2} &\left( \wt{\cU}^* (t;0) - \cU_\infty^* (t;0) \right) a^* (\wt{O}_c \wt{\ph}_t)^{j} a (\wt{O}_c \wt{\ph}_t)^{2m-k-j} \cU_\infty (t;0) \Omega \Big\| \\ \leq \; & C e^{Kt} (N^{-1/2} + \alpha_N (N+1)^2) \\ & \times \left\| \left( \cK + \cN^{20r-4} + 1 \right)^{1/2} (\cN+1)^{3} a^* (\wt{O}_c \wt{\ph}_t)^{j} a (\wt{O}_c \wt{\ph}_t)^{2m-k-j} \cU_\infty (t;0) \Omega \right\| \end{split} \]
Using Lemma \ref{lm:phiNK}, we conclude that
\[ \begin{split}
\Big\| &(\cN+1)^{1/2} \left( \wt{\cU}^* (t;0) - \cU_\infty^* (t;0) \right) a^* (\wt{O}_c \wt{\ph}_t)^{j} a (\wt{O}_c \wt{\ph}_t)^{2m-k-j} \cU_\infty (t;0) \Omega \Big\| \\ &\leq C e^{Kt} (N^{-1/2} + \alpha_N (N+1)^2) \, \| \wt{O}_c \wt{\ph}_t \|_{H^1}^{2m-k} \, \left\| (\cK+\cN^{20r-4}+ 1)^{1/2} (\cN+1)^{3+ m -k/2} \cU_\infty (t;0) \Omega \right\| \\
&\leq C e^{Kt} (N^{-1/2} + \alpha_N (N+1)^2) \, \| \wt{O}_c \wt{\ph}_t \|_{H^1}^{2m-k} \\ &\hspace{1cm} \times \left[ \left\| \cK \cU_\infty (t;0) \Omega \right\| + \left\| (\cN+1)^{6+ 2m -k} \cU_\infty (t;0) \Omega \right\| + \left\|(\cN+1)^{10r+1 + m -k/2} \cU_\infty (t;0) \Omega \right\| \right]
\end{split} \]
where, in the last inequality we used Cauchy-Schwarz to separate $K^{1/2}$ and $(\cN+1)^{3+m-k/2}$. {F}rom Proposition \ref{prop:2}, we obtain
\[ \begin{split}
\Big\| (\cN+1)^{1/2} \left( \wt{\cU}^* (t;0) - \cU_\infty^* (t;0) \right) &a^* (\wt{O}_c \wt{\ph}_t)^{j} a (\wt{O}_c \wt{\ph}_t)^{2m-k-j} \cU_\infty (t;0) \Omega \Big\|  \\
&\leq C e^{Kt} (N^{-1/2} + \alpha_N (N+1)^2) \, \| \wt{O}_c \wt{\ph}_t \|_{H^1}^{2m-k}  \\
&\leq C e^{Kt} (N^{-1/2} + \alpha_N (N+1)^2) \, \left( \| \nabla \, O (1-\Delta)^{-1/2} \| + \| O \| \right)^{2m-k}
\end{split} \]
Since $\alpha_N = N^{-r}$ for $r > 3 + k/2$, the r.h.s. of the last equation converges to zero, as $N \to \infty$. To estimate the second term in (\ref{eq:compa3}), we observe that, by Proposition \ref{prop:1} and Proposition \ref{prop:3},
\[ \begin{split}
\Big\| (\cN+1)^{1/2} \cU (t;0) \, a^* (&\wt{O}_c \wt{\ph}_t)^j  a (\wt{O}_c \wt{\ph}_t)^{2m-k-j} \left(\wt{\cU} (t;0) -\cU_\infty (t;0) \right) \Omega \Big\| \\ \leq \; & C e^{Kt} \left\| (\cN+1)^{2} \, a^* (\wt{O}_c \wt{\ph}_t)^j  a (\wt{O}_c \wt{\ph}_t)^{2m-k-j} \left(\wt{\cU} (t;0) -\cU_\infty (t;0) \right) \Omega \right\|
\\ \leq \; & C e^{Kt} \| \wt{O}_c \wt{\ph}_t \|^{2m-k} \, \left\| (\cN+1)^{2+m-k/2} \, \left(\wt{\cU} (t;0) -\cU_\infty (t;0) \right) \Omega \right\| \\
\leq \; & C e^{Kt} \left( N^{-1/2} + \alpha_N (N+1)^{3+m-k/2} \right) \, \| O \|^{2m-k}
\end{split} \]
which converges to zero as $N \to \infty$, because $\alpha_N = N^{-r}$, with $r > 3 + k/2 \geq 3+m-k/2$. {F}rom (\ref{eq:compa3}), using also (\ref{eq:Nxi}), we conclude that
\begin{equation}\label{eq:higher2b} \begin{split}
\E_{\wt{\psi}_{N,t}} \wt{\cO}_t^k = \; & \sum_{m= k/2}^{k} \frac{1}{m!} {m \choose 2m - k} \frac{k!}{2^{k-m}}  \, \langle \wt{\ph}_t , \wt{O}_c^2 \wt{\ph}_t \rangle^{k -m} \\ &\hspace{2cm} \times  \left\langle
\xi_N , \cU_{\infty}^* (t;0) : \left( a (\wt{O}_c \wt{\ph}_t) + a^* (\wt{O}_c \wt{\ph}_t) \right)^{2m-k} : \cU_{\infty} (t;0) \Omega \right\rangle \\ &+ o(1)
\end{split}
\end{equation}
as $N \to \infty$.

Suppose now that $k= 2\ell+1$ is odd. Expanding $: (a (O_c \ph_t) + a^* (O_c \ph_t))^{2m-2\ell-1}:$ we find a sum of (normally ordered) terms of the form $(a^* (O_c \ph_t))^p (a (O_c \ph_t))^q$ with $p+q = 2m-2\ell-1$. Conjugating with the limiting dynamics $\cU_\infty (t;0)$, and using that
\[ \begin{split}
\cU^*_\infty (t;0) a (O_c \ph_t) \cU_{\infty} (t;0) &= a (F_t) + a^* (G_t) \\
\cU^*_\infty (t;0) a^* (O_c \ph_t) \cU_{\infty} (t;0) &= a^* (F_t) + a (G_t)
\end{split} \]
with $F_t = U (t;0) O_c \ph_t$ and $G_t = J V (t;0) O_c \ph_t$, we find that $\cU^*_\infty (t;0) : (a (O_c \ph_t) + a^* (O_c \ph_t))^{2m-2\ell-1}: \cU_\infty (t;0)$ is a linear combination of terms of the form $(a^* (F_t) + a (G_t))^p (a (F_t) + a^* (G_t))^q$ with $p+q =2m -2 \ell -1$. Expanding the parenthesis, we conclude that it
is a linear combination of terms given by (not necessarily normally ordered) products of creation operators $a^* (F_t)$ and $a^* (G_t)$ and annihilation operators $a(F_t)$ and $a(G_t)$. In every term the total number of annihilation and creation operators is $2m -2\ell -1$. Consider a single term, with a total of $\wt{p}$ creation operators, and of $\wt{q} = 2m-2\ell -1-\wt{p}$ annihilation operators. Since these operators act on the vacuum, this term is a vector in the sector of $\cF$ with exactly $\wt{p}-\wt{q} = 2\wt{p} -2m + 2\ell +1$ particles. On the other hand, from Proposition \ref{prop:xi}, we observe that the odd components of the vector $\xi_N = \{ \xi_N^{(0)} , \xi_N^{(1)}, \dots \} \in \cF$ vanish in the limit $N \to \infty$. Hence, we conclude that, for any fixed $\ell \in \bN$,
\[ \E_{\wt{\psi}_{N,t}} \wt{\cO}_t^{2\ell+1} \to 0 \]
as $N \to \infty$.

Next, we consider the even moments. For $k=2\ell$, (\ref{eq:higher2b}) gives (changing the summation variable to $m-\ell$)
\begin{equation}\label{eq:higher3} \begin{split}
\E_{\wt{\psi}_{N,t}} \wt{\cO}_t^{2\ell} = \; &\sum_{m=0}^\ell \frac{2\ell!}{2m! \ell - m! 2^{\ell-m}} \langle \ph_t, O_c^2 \ph_t \rangle^{\ell-m} \left\langle \cU_\infty (t;0) \xi_N , : \left( a^* (O_c \ph_t) + a (O_c \ph_t) \right)^{2m} : \cU_\infty (t;0) \Omega \right\rangle \\ &+ o(1)
\end{split}
\end{equation}
as $N \to \infty$. {F}rom Proposition \ref{prop:normal-sum}, we find
\[ \begin{split}
\E_{\wt{\psi}_{N,t}} \wt{\cO}_t^{2\ell} = \; & \left\langle \cU_\infty (t;0) \xi_N , \left( a^* (O_c \ph_t) + a (O_c \ph_t) \right)^{2\ell} \cU_\infty (t;0) \Omega \right\rangle + o(1)
\end{split} \]
With the notation $A(f,g) = a(f) + a^* (\overline{g})$, and using Theorem \ref{thm:bog}, we find
\begin{equation}\label{eq:O2l} \begin{split}
\E_{\wt{\psi}_{N,t}} \, \wt{\cO}_t^{2\ell} = \;& \left\langle \cU_\infty (t;0) \xi_N , A \left( O_c \ph_t , JO_c \ph_t \right)^{2\ell}\cU_\infty (t;0) \Omega \right\rangle + o(1) \\
= \;& \left\langle \xi_N , A \left( U (t;0) O_c \ph_t + J V (t;0) O_c \ph_t , V (t;0) O_c \ph_t + J U (t;0) O_c \ph_t \right)^{2\ell} \Omega \right\rangle \end{split} \end{equation}
{F}rom Proposition \ref{prop:xi}, we write
\[ \xi_N = \sum_{j \geq 0} \xi_N^{(j)} a^* (\ph)^j \, \Omega \, . \]
Using Lemma \ref{lm:a*Omega} to expand $ A \left( U (t;0) O_c \ph_t + J V (t;0) O_c \ph_t , V (t;0) O_c \ph_t + J U (t;0) O_c \ph_t \right)^{2\ell} \Omega$, we find
\[ \begin{split}
\E_{\wt{\psi}_{N,t}} \, \wt{\cO}_t^{2\ell}
= \;& \sum_{j=0}^{2\ell} \xi_N^{(j)} \sum_{m=0}^\ell \frac{2\ell!}{2m! \ell -m ! 2^{\ell-m}} \| U (t;0) O_c \ph_t + J V (t;0) O_c \ph_t \|^{2(\ell-m)} \\ & \hspace{4cm}\times \langle a^* (\ph)^j \Omega , a^* (U (t;0) O_c \ph_t + J V (t;0) O_c \ph_t)^{2m} \Omega \rangle
\end{split} \]
Since, for any $f,g \in L^2 (\bR^3)$, and for any $j,k \in \bN$,
\[ \langle a^* (f)^j \Omega, a^* (g)^k \Omega \rangle = \delta_{jk} (f,g)^k k! \]
we conclude that
\[ \begin{split}
\E_{\wt{\psi}_{N,t}} & \, \wt{\cO}_t^{2\ell} \\
= \;& \sum_{m=0}^\ell \xi_N^{(2m)} \, \frac{2\ell!}{\ell -m ! 2^{\ell-m}} \| U (t;0) O_c \ph_t + J V (t;0) O_c \ph_t \|^{2(\ell-m)} \langle \ph , U (t;0) O_c \ph_t + J V (t;0) O_c \ph_t \rangle^{2m} \\
= \; & 2\ell! \sum_{m=0}^\ell \frac{\| U (t;0) O_c \ph_t + J V (t;0) O_c \ph_t \|^{2(\ell-m)} \langle \ph, U (t;0) O_c \ph_t + J V (t;0) O_c \ph_t \rangle^{2m}}{2^{\ell -m} (\ell -m)!} \xi_N^{(2m)}  +o(1)
\end{split} \]
{F}rom Proposition \ref{prop:xi}, we find that, for any $m=0,\dots,\ell$ fixed,
\[ \xi^{(2m)}_N \to \frac{(-1)^m}{2^m m!} \]
as $N \to \infty$. Hence
\begin{equation}\label{eq:last} \begin{split}
\lim_{N\to \infty} &\E_{\wt{\psi}_{N,t}} \, \wt{\cO}_t^{2\ell}
\\ = \; & \frac{2\ell!}{\ell! 2^\ell} \sum_{m=0}^\ell (-1)^m  {\ell \choose m} \, \| U (t;0) O_c \ph_t + J V (t;0) O_c \ph_t \|^{2(\ell-m)} \langle \ph, U (t;0) O_c \ph_t + J
V_t O_c \ph_t \rangle^{2m} \\ = \; & \frac{2\ell!}{2^{\ell} \ell!} \left( \| U (t;0) O_c \ph_t + J V (t;0) O_c \ph_t \|^2 - \langle  \ph, U (t;0) O_c \ph_t + J V (t;0) O_c \ph_t \rangle^{2} \right)^{2\ell} \\
%= \; & \frac{2\ell!}{2^{\ell} \ell!} \left( \| U (t;0) O_c \ph_t + J V (t;0) %O_c \ph_t \|^2 - |\langle  \ph, U (t;0) O_c \ph_t + J V (t;0) O_c \ph_t %\rangle|^{2} \right)^{2\ell}
%\\
%= \; & \frac{2\ell!}{2^{\ell} \ell!} \left( \| U (t;0) O \ph_t + J V (t;0) O \ph_t %\|^2 - |\langle  \ph, U (t;0) O \ph_t + J V (t;0) O \ph_t \rangle|^{2} %\right)^{2\ell}
\end{split} \end{equation}
Now, we observe that
\begin{equation}\label{eq:var5} \begin{split} -i \left\langle (\ph , - \overline{\ph}), \Theta (t;0) \left(O_c \ph_t , J O_c \ph_t \right) \right\rangle = &- i \langle \ph , U (t;0) O_c \ph_t + J V (t;0) O_c \ph_t \rangle \\ &+ i \langle \overline{\ph} , V (t;0) O_c \ph_t + J U (t;0) O_c \ph_t \rangle \\ = & 2 \text{Im }
\langle \ph, U (t;0) O_c \ph_t + J V (t;0) O_c \ph_t \rangle \end{split} \end{equation}
{F}rom (\ref{eq:Ttph}), and since $\Theta (t;0)^* = S \Theta_t
^{-1} S$ (with $S$ as defined in (\ref{eq:bog22})), we conclude that
\begin{equation}\label{eq:im0} \begin{split} 2 \text{Im }
\langle \ph, U (t;0) O_c \ph_t + J V (t;0) O_c \ph_t \rangle = & \;
-i \left\langle \Theta (t;0)^* (\ph , - \overline{\ph}), \left(O_c \ph_t , J O_c \ph_t \right) \right\rangle \\ = &\; -i \left\langle \Theta (t;0)^{-1} (\ph , \overline{\ph}), \left(O_c \ph_t , - J O_c \ph_t \right) \right\rangle \\ = &-i \left\langle (\ph_t , \overline{\ph}_t), \left(O_c \ph_t , - JO_c \ph_t \right) \right\rangle =  2 \text{Im } \langle \ph_t , O_c \ph_t \rangle = 0 \end{split} \end{equation}
Thus
\[ \langle \ph, U (t;0) O_c \ph_t + J V (t;0) O_c \ph_t \rangle = \text{Re } \langle \ph, U (t;0) O_c \ph_t + J V (t;0) O_c \ph_t \rangle \]
and, from (\ref{eq:last}),
\begin{equation}\label{eq:var6}
\begin{split}
\lim_{N\to \infty} \E_{\wt{\psi}_{N,t}} \, \wt{\cO}_t^{2\ell}
= \; &\frac{2\ell!}{2^{\ell} \ell!} \left( \| U (t;0) O_c \ph_t + J V (t;0) O_c \ph_t \|^2 - |\langle  \ph, U (t;0) O_c \ph_t + J V (t;0) O_c \ph_t \rangle|^{2} \right)^{2\ell} \\
= \; & \frac{2\ell!}{2^{\ell} \ell!} \; \sigma_{c,t}^{4\ell}
\end{split}
\end{equation}
with \[ \sigma^2_{c,t} = \| U (t;0) O_c \ph_t + J V (t;0) O_c \ph_t \|^2 - |\langle  \ph, U (t;0) O_c \ph_t + J V (t;0) O_c \ph_t \rangle|^{2} \]
Finally, we note that
\[ \begin{split} &\left\langle (\ph, \overline{\ph}), \Theta (t;0) \left(O_c \ph_t , J O_c \ph_t \right) \right\rangle \\ & \hspace{3cm}= \left\langle (\ph, \overline{\ph}), ( U (t;0) O_c\ph_t + J V (t;0) O_c \ph_t , JU (t;0) O_c \ph_t + V (t;0) O_c \ph_t) \right\rangle \\  & \hspace{3cm}= 2 \text{Re } \langle \ph , U (t;0) O_c \ph_t + J V (t;0) O_c \ph_t \rangle \end{split} \]
and that, similarly,
\[ \| U (t;0) O_c \ph_t + J V (t;0) O_c \ph_t \|^2 = \frac{1}{2} \left \langle \Theta (t;0) \left(O_c \ph_t , JO_c \ph_t \right) , \Theta (t;0) \left(O_c \ph_t , JO_c \ph_t \right) \right\rangle \]
This gives
\begin{equation}\label{eq:var7}\begin{split}
\sigma_{c,t}^2 = &\; \frac{1}{2} \left[  \left\langle \Theta (t;0) \left(O_c \ph_t , JO_c \ph_t\right) , \Theta (t;0) \left(O_c \ph_t , O_c \ph_t \right) \right\rangle \right.\\ & \hspace{1cm} \left.- \left\langle \Theta (t;0) \left(O_c \ph_t , JO_c \ph_t\right) , \frac{1}{\sqrt{2}} (\ph, \overline{\ph}) \right\rangle \left\langle \frac{1}{\sqrt{2}} (\ph, \overline{\ph}),  \Theta (t;0) \left(O_c \ph_t , JO_c \ph_t\right) \right\rangle \right]
\end{split} \end{equation}
Writing $O_c = O - \E_{\ph_t} O = O - \langle \ph_t, O \ph_t \rangle$, we easily find
\begin{equation}\label{eq:rem-tilda}\begin{split} \sigma_{c,t}^2 = \;& \frac{1}{2} \left[  \left\langle \Theta (t;0) \left(O \ph_t , J O \ph_t \right) , \Theta (t;0) \left( O \ph_t , J O \ph_t \right) \right\rangle \right. \\ & \hspace{1cm} \left. - \left\langle \Theta (t;0) \left(O \ph_t , J O \ph_t \right) , \frac{1}{\sqrt{2}} (\ph, \overline{\ph}) \right\rangle \left\langle \frac{1}{\sqrt{2}} (\ph, \overline{\ph}),  \Theta (t;0) \left(O \ph_t , J O \ph_t \right) \right\rangle \right] \\
= \; &\| U (t;0) O \ph_t + J V (t;0) O \ph_t \|^2 - \left| \langle \ph, U (t;0) O \ph_t + J V (t;0) O \ph_t \rangle \right|^2 \\ = \; & \sigma_t^2 \end{split} \end{equation}
With (\ref{eq:var6}), this concludes the proof of Lemma \ref{lm:mom-reg}.

\section{Combinatorial estimates}
\label{sec:combi}
\setcounter{equation}{0}

In this section, we collect some technical results, used in Section \ref{sec:CLT}, to handle the combinatorial problems.

The next proposition is used to determine the properties of the Fock vector $\xi_N$ defined in (\ref{eq:xi-def}) in the limit of large $N$. We use here an important observation from \cite{CL}.
\begin{proposition}\label{prop:xi}
Let
\[ \xi_N = d_N W^* (\sqrt{N} \ph) \frac{a^* (\ph)^N}{\sqrt{N!}} \Omega \,. \]
with $d_N = e^{N/2} \sqrt{N!} \, N^{-N/2}$. Then
\begin{equation}\label{eq:xi1} \xi_N = \sum_{\ell=0}^{\infty} \xi_N^{(\ell)} a^* (\ph)^\ell \Omega \end{equation}
with
\begin{equation}\label{eq:xi2} \xi_N^{(\ell)} = \sum_{j=0}^\ell (-1)^j N^{j-\ell/2} \frac{N!}{N-\ell+j! \ell-j! j!} \end{equation}
Moreover, for any fixed $m \in \bN$, we have $\xi_N^{(2m+1)} \to 0$
and
\[\xi_N^{(2m)} \to \frac{(-1)^m}{2^m m!} \]
as $N \to \infty$.
\end{proposition}

\begin{proof}
We start from the definition of $\xi_N$, given by
\[ \xi_N = d_N W^* (\sqrt{N} \ph) \frac{a^* (\ph)^N}{\sqrt{N!}} \Omega \]
with $d_N = e^{N/2} \sqrt{N!} N^{-N/2} \simeq N^{1/4}$. Using the properties of the Weyl operator $W^* (\sqrt{N} \ph)$, we obtain
\[ \begin{split} \xi_N &= d_N \frac{(a^* (\ph) + \sqrt{N})^N}{\sqrt{N!}} W^* (\sqrt{N} \ph) \Omega = d_N \frac{(a^* (\ph) + \sqrt{N})^N}{\sqrt{N!}} e^{-N/2} \sum_{j=0}^{\infty} (-1)^j N^{j/2} \frac{a^* (\ph)^j}{j!} \Omega \end{split} \]
The restriction of $\xi_N$ on the $\ell$-particle sector is given by
\[ \begin{split} \xi_N |_{\cF_\ell} &= \frac{d_N e^{-N/2}}{\sqrt{N!}} \sum_{j=0}^\ell \frac{(-1)^j N^{j/2}}{j!} {N \choose \ell-j} N^{\frac{N-\ell+j}{2}} a^* (\ph)^\ell \Omega
\\ &=  \left[ \sum_{j=0}^\ell (-1)^j N^{j-\ell/2} \frac{N!}{N-\ell+j! \ell-j! j!} \right] a^* (\ph)^\ell \Omega \\ & = \xi_N^{(\ell)} a^* (\ph)^\ell \Omega \end{split} \]
This proves (\ref{eq:xi1}). Now we study the asymptotics of $\xi_N^{(\ell)}$, as $N \to \infty$. To this end we recognize, following \cite{CL}, that
\[ \xi_N^{(\ell)} = N^{-\ell/2} L_\ell^{(N-\ell)} (N) \]
where
\[ L_\ell^{(\alpha)} (x) = \sum_{j=0}^\ell (-1)^j {\ell + \alpha
\choose \ell - j} \frac{x^j}{j!} \]
are the generalized Laguerre polynomials of dergree $\ell$, which satisfy the recursion relations (see, for example, \cite{CL})
\[ L_\ell^{(\alpha)} (x) = \frac{\alpha+1 -x}{\ell} \, L_{\ell-1}^{(\alpha+1)} (x) - \frac{x}{\ell} \, L_{\ell-2}^{(\alpha+2)} (x) \]
This gives the recursion
\begin{equation}\label{eq:rec-xi} \xi_N^{(\ell)} = \frac{1-\ell}{\ell} \, N^{-1/2} \, \xi_N^{(\ell-1)} - \frac{1}{\ell} \, \xi_N^{(\ell-2)} \end{equation}
A simple computation shows that $\xi_N^{(0)} = 1$ and $\xi^{(1)}_N = 0$.
Eq. (\ref{eq:rec-xi}) implies therefore that, for any fixed $\ell \in \bN$, $\xi^{(\ell)}_N \to 0$ if $\ell$ is odd, and \[ \xi^{(\ell)}_N \to \frac{(-1)^\ell}{\ell (\ell-2) \dots 2} = \frac{(-1)^\ell}{2^{\ell/2} (\ell/2)!} \]
if $\ell$ is even.
\end{proof}

The next lemma is used to expand \[ \begin{split}  A &\left( U (t;0) O_c \ph_t + J V (t;0) O_c \ph_t , V (t;0) O_c \ph_t + J U (t;0) O_c \ph_t \right)^{2\ell} \Omega \\ &\hspace{2cm} = \left( a^* (U (t;0) O_c \ph_t + J V (t;0) O_c \ph_t) + a (U (t;0) O_c \ph_t + J V (t;0) O_c \ph_t) \right)^{2\ell} \Omega \end{split} \] appearing in (\ref{eq:O2l}).

\begin{lemma}\label{lm:a*Omega}
For any $F \in L^2 (\bR^3)$, we have
\[(a(F) + a^* (F))^{2\ell} \Omega = \sum_{m=0}^\ell \frac{2\ell!}{2m! \ell -m ! 2^{\ell-m}} \| F \|^{2(\ell-m)} a^* (F)^{2m} \Omega \]
\end{lemma}

\begin{proof}
We proceed by induction over $\ell$. The statement is clearly correct for $\ell =0$. Assume it is correct for a fixed $\ell \in \bN$. We prove it holds true for $\ell+1$. To this end we write, using the induction assumption,
\begin{equation}\label{eq:a*1} \begin{split} (a(F) + a^* (F))^{2\ell+2} \Omega &= (a(F) + a^* (F))^2 \, (a(F) + a^*(F))^{2\ell} \Omega \\ &= \sum_{m=0}^\ell \frac{2\ell!}{2m! \ell-m! 2^{\ell-m}} \| F \|^{2 (\ell-m)} (a (F) + a^* (F))^2 a^* (F)^{2m} \Omega
\end{split} \end{equation}
Next, we observe that, using the canonical commutation relations,
\[ \begin{split} (a (F) + a^* (F))^2 a^* (F)^{2m} = \; &a^* (F)^{2(m+1)} + a (F)^2 a^* (F)^{2m} + 2 a(F) a^* (F)^{2m+1} - \| F \|^2 a^* (F)^{2m} \\ = \;& a^* (F)^{2(m+1)} + 2m \| F \|^2 \,  a(F) a^* (F)^{2m-1} + (4m+1) \| F \|^2 a^* (F)^{2m} \\ = \;& a^* (F)^{2(m+1)} + 2m (2m-1) \| F \|^4 \, a^* (F)^{2(m-1)} + (4m+1) \| F \|^2 \, a^* (F)^{2m} \end{split} \]
Inserting in (\ref{eq:a*1}), we find
\[ \begin{split}
(a(F) + a^* (F))^{2\ell+2} \Omega =\; & \sum_{m=0}^\ell \frac{2\ell!}{2m! \ell-m! 2^{\ell-m}} \| F \|^{2 (\ell-m)} a^* (F)^{2 (m+1)} \Omega \\ &+  \sum_{m=0}^\ell \frac{2\ell! (4m+1)}{2m! \ell-m! 2^{\ell-m}} \| F \|^{2 (\ell+1-m)} a^* (F)^{2m} \Omega \\ &+  \sum_{m=1}^\ell \frac{2\ell!}{2(m-1)! \ell-m! 2^{\ell-m}} \| F \|^{2 (\ell+2-m)} a^* (F)^{2(m-1)} \Omega
\end{split} \]
Changing the summation index to $m+1$ in the first and, respectively, to $m-1$ in the third term on the r.h.s. of the last equation, we obtain
\[ \begin{split}
(a(F) &+ a^* (F))^{2\ell+2} \Omega \\ =\; & \sum_{m=1}^{\ell+1} \frac{2\ell!}{2(m-1)! \ell+1-m! 2^{\ell+1-m}} \| F \|^{2 (\ell+1-m)} a^* (F)^{2m} \Omega \\ &+  \sum_{m=0}^\ell \frac{2\ell! (4m+1)}{2m! \ell-m! 2^{\ell-m}} \| F \|^{2 (\ell+1-m)} a^* (F)^{2m} \Omega \\ &+  \sum_{m=0}^{\ell-1} \frac{2\ell!}{2m! \ell-1-m! 2^{\ell-1-m}} \| F \|^{2 (\ell+1-m)} a^* (F)^{2m} \Omega \\ = \; &a^* (F)^{2(\ell+1)} \Omega + (\ell (2\ell-1)+ 4\ell+1) \| F \|^2 a^* (F)^{2\ell} \Omega \\ &+ \sum_{m=1}^{\ell-1} \frac{2\ell!}{2m! \ell+1-m! 2^{\ell+1-m}} \| F \|^{2(\ell+1-m)} \\ &\hspace{1cm} \times \left(2m (2m-1) +2 (4m+1)(\ell+1-m) + (\ell+1-m)(\ell-m) \right) \, a^* (F)^{2m} \Omega \\
= \; & \sum_{m=0}^{\ell+1} \frac{2(\ell+1)!}{2m! \ell+1 -m! 2^{\ell+1-m}} \| F \|^{2 (\ell+1-m)} a^* (F)^{2m} \Omega \, .
\end{split} \]
\end{proof}

Finally, the next proposition is used to evaluate a certain sum of normally ordered monomials in $a(O_c \ph_t)$ and $a^* (O_c \ph_t)$, appearing in the computation of high moments of the random variable $\wt{\cO}_t$.

\begin{proposition} \label{prop:normal-sum}
For any $\ell \in \bN$, $F \in L^2 (\bR^3)$, $\psi_1, \psi_2 \in \cF$, we have
\[ \sum_{m=0}^\ell \frac{2\ell!}{2m! \ell- m! 2^{\ell-m}}  \, \| F \|^{2(\ell-m)}  \, \left\langle \psi_1, : \left( a (F) + a^* (F) \right)^{2m} : \psi_2 \right\rangle = \left\langle \psi_1 , \left( a (F) + a^* (F) \right)^{2\ell} \psi_2 \right\rangle \]
\end{proposition}

\begin{proof}
The proof proceeds by induction over $\ell \in \bN$. For $\ell = 0$, the claim is clearly correct. We assume now it holds for $\ell-1$, and we prove it for $\ell$. To this end, we use that, for every $k \geq 1$,
\begin{equation}\label{eq:ide} (a (F) + a^* (F)) : (a (F) + a^* (F))^k : = :(a(F) + a^* (F))^{k+1} : + k \| F \|^2 \, : (a(F) + a^* (F))^{k-1}: \end{equation}
Applying the last identity twice, we conclude that
\begin{equation}\label{eq:norm1} \begin{split}
\langle \psi_1 , : (a(F) +a^* (F))^{2\ell} : \psi_2 \rangle = \; & \langle \psi_1, (a(F) +a^* (F))^2 : (a (F) + a^* (F))^{2(\ell-1)}: \psi_2 \rangle \\ &-(4\ell-3) \langle \psi_1, : (a(F) + a^* (F))^{2(\ell-1)}: \psi_2 \rangle \\ &- (2\ell-2)(2\ell-3) \langle \psi_1, :(a(F) + a^* (F))^{2(\ell-2)}: \psi_2 \rangle \end{split} \end{equation}
Using the induction assumption, we find
\[ \begin{split} \langle \psi_1 , (a(F)+a^* (F))^2 &:(a(F)+a^* (F))^{2(\ell-1)} : \psi_2 \rangle \\ =\; & \langle \psi_1, (a (F) + a^* (F))^{2\ell} \psi_2 \rangle \\ &- \sum_{m=0}^{\ell-2} \frac{2(\ell-1)! \| F \|^{2(\ell-1-m)}}{2m! \ell-1-m! 2^{\ell-1-m}} \langle \psi_1, (a (F) +a^* (F))^2 :(a(F) + a^* (F))^{2m} : \psi_2 \rangle\end{split} \]
Again, using (\ref{eq:ide}), we obtain
\[ \begin{split}
\langle \psi_1 , (a(F)+a^* (F))^2 &:(a(F)+a^* (F))^{2(\ell-1)} : \psi_2 \rangle \\ = \; &\langle \psi_1, (a (F) + a^* (F))^{2\ell} \psi_2 \rangle \\ &- \sum_{m=0}^{\ell-2} \frac{2(\ell-1)! \| F \|^{2(\ell-1-m)}}{2m! \ell-1-m! 2^{\ell-1-m}} \langle \psi_1, :(a(F) + a^* (F))^{2(m+1)} : \psi_2 \rangle  \\ &-\sum_{m=0}^{\ell-2} \frac{2(\ell-1)! (4m+1) \| F \|^{2(\ell-m)}}{2m! \ell-1-m! 2^{\ell-1-m}} \langle \psi_1, :(a(F) + a^* (F))^{2m} : \psi_2 \rangle \\ &-\sum_{m=1}^{\ell-2} \frac{2(\ell-1)! \| F \|^{2(\ell+1-m)}}{2(m-1)! \ell-1-m! 2^{\ell-1-m}} \langle \psi_1, :(a(F) + a^* (F))^{2(m-1)} : \psi_2 \rangle
\end{split}\]
Combining the last equation with (\ref{eq:norm1}), we deduce that
\[ \begin{split}
\langle \psi_1 , : (a(F) +a^* (F)&)^{2\ell} : \psi_2 \rangle \\ = \; &\langle \psi_1, (a (F) + a^* (F))^{2\ell} \psi_2 \rangle \\ &- \sum_{m=0}^{\ell-2} \frac{2(\ell-1)! \| F \|^{2(\ell-1-m)}}{2m! \ell-1-m! 2^{\ell-1-m}} \langle \psi_1, :(a(F) + a^* (F))^{2(m+1)} : \psi_2 \rangle  \\ &-\sum_{m=0}^{\ell-1} \frac{2(\ell-1)! (4m+1) \| F \|^{2(\ell-m)}}{2m! \ell-1-m! 2^{\ell-1-m}} \langle \psi_1, :(a(F) + a^* (F))^{2m} : \psi_2 \rangle \\ &-\sum_{m=1}^{\ell-1} \frac{2(\ell-1)! \| F \|^{2(\ell+1-m)}}{2(m-1)! \ell-1-m! 2^{\ell-1-m}} \langle \psi_1, :(a(F) + a^* (F))^{2(m-1)} : \psi_2 \rangle
\end{split} \]
Changing summation variables, we find
\[ \begin{split}
\langle \psi_1 , &:(a(F)+a^* (F))^{2\ell} : \psi_2 \rangle \\ = \; &\langle \psi_1, (a (F) + a^* (F))^{2\ell} \psi_2 \rangle \\ &- \sum_{m=1}^{\ell-1} \frac{2(\ell-1)! \| F \|^{2(\ell-m)}}{2(m-1)! \ell-m! 2^{\ell-m}} \langle \psi_1, :(a(F) + a^* (F))^{2m} : \psi_2 \rangle  \\ &-\sum_{m=0}^{\ell-1} \frac{2(\ell-1)! (4m+1) \| F \|^{2(\ell-m)}}{2m! \ell-1-m! 2^{\ell-1-m}} \langle \psi_1, :(a(F) + a^* (F))^{2m} : \psi_2 \rangle \\ &-\sum_{m=0}^{\ell-2} \frac{2(\ell-1)! \| F \|^{2(\ell-m)}}{2m! \ell-2-m! 2^{\ell-2-m}} \langle \psi_1, :(a(F) + a^* (F))^{2m} : \psi_2 \rangle \\  = \; &\langle \psi_1, (a (F) + a^* (F))^{2\ell} \psi_2 \rangle \\ &- \left((\ell-1)(2\ell-3) +4\ell-3 \right) \| F \|^2 \langle \psi_1, : (a(F) +a^* (F))^{2(\ell-1)}: \psi_2 \rangle \\ &- \left(\frac{2(\ell-1)! }{(\ell-1)! 2^{\ell-1}} + \frac{2(\ell-1)!}{(\ell-2)!2^{\ell-2}} \right) \| F \|^{2\ell} \langle \psi_1, \psi_2 \rangle \\
&-\sum_{m=1}^{\ell-2} \frac{2(\ell-1)! \| F \|^{2(\ell-m)}}{2m! \ell-m! 2^{\ell-m}}\, \langle \psi_1, : (a(F) + a^* (F))^{2m} : \psi_2 \rangle \\ &\hspace{2cm} \times \left( 2m (2m-1) + 2(4m+1)(\ell-m)+4(\ell-m)(\ell-m-1) \right)
\end{split}
\]
Since \[ (\ell-1) (2\ell-3) +4\ell-3 = \ell (2\ell-1) = \frac{2\ell!}{(2(\ell-1)! 2}, \]
\[  \frac{2(\ell-1)! }{(\ell-1)! 2^{\ell-1}} + \frac{2(\ell-1)!}{(\ell-2)!2^{\ell-2}} = \frac{2\ell!}{\ell! 2^\ell} \] and
\[ 2m (2m-1) + 2(4m+1)(\ell-m)+4(\ell-m)(\ell-m-1) = 2\ell (2\ell-1) \]
we conclude that
\[ \begin{split}
\langle \psi_1 , :(a(F)+a^* (F))^{2\ell} : \psi_2 \rangle = \; &\langle \psi_1, (a (F) + a^* (F))^{2\ell} \psi_2 \rangle \\
&-\sum_{m=0}^{\ell-1} \frac{2\ell! \| F \|^{2(\ell-m)}}{2m! \ell-m! 2^{\ell-m}} \langle \psi_1, : (a(F) + a^* (F))^{2m} : \psi_2 \rangle
\end{split}
\]
which implies
\[ \langle \psi_1, (a (F) + a^* (F))^{2\ell} \psi_2 \rangle = \sum_{m=0}^\ell \frac{2\ell! \| F \|^{2(\ell-m)}}{2m! \ell-m! 2^{\ell-m}} \langle \psi_1, : (a(F) + a^* (F))^{2m} : \psi_2 \rangle \]
and concludes the proof of the proposition.
\end{proof}

{\it Acknowledgements.} We would like to thank Ji Oon Lee for useful discussions and, in particular, for proposing a new proof of Proposition \ref{prop:xi}, substantially simpler than the previous one. We are also grateful to Jan Derezi{\'n}ski for his help on questions related to Bogoliubov transformations.

\end{document}